\pdfoutput=1
\documentclass{vldb}

\usepackage{type1cm}     
\usepackage{graphicx}     
\usepackage{xspace}     
\usepackage{amsmath, amssymb}	
\usepackage{balance}     
\usepackage{booktabs}     
\usepackage[bf,tableposition=top]{caption}     
\usepackage[hyphens]{url}     
\usepackage{bold-extra}     
\usepackage{siunitx}          
\usepackage[vlined,linesnumbered,ruled,noend]{algorithm2e}     
\usepackage[square,numbers,sort]{natbib}     
\usepackage{soul}
\usepackage[show]{chato-notes}
\usepackage{units}
\usepackage[dvipsnames,svgnames,hyperref]{xcolor}
\usepackage{multirow}

\graphicspath{{fig/}}
\SetSideCommentRight

\definecolor{myred}{HTML}{AF2525}
\newcommand\txtred[1]{{\color{myred}{#1}}}

\definecolor{mygreen}{HTML}{2A8F56}
\newcommand\txtgreen[1]{{\color{mygreen}{#1}}}

\definecolor{myblue}{named}{Blue}

\newcommand{\spara}[1]{\smallskip\noindent\textbf{#1}}

\newenvironment {squishlist}
{\begin{list}{$\bullet$}
  { \setlength{\itemsep}{0pt}
     \setlength{\parsep}{3pt}
     \setlength{\topsep}{3pt}
     \setlength{\partopsep}{0pt}
     \setlength{\leftmargin}{1.5em}
     \setlength{\labelwidth}{1em}
     \setlength{\labelsep}{0.5em} } }
{\end{list}}

\newtheorem{problem}{Problem}

\newtheorem{theorem}{Theorem}

\newcommand{\bigOh}{\ensuremath{\mathcal{O}}\xspace}
\newcommand{\simop}{\ensuremath{\mathrm{sim}}\xspace}
\newcommand{\dotop}{\ensuremath{\mathrm{dot}}\xspace}
\newcommand{\stream}{\ensuremath{\mathcal{S}}\xspace}
\newcommand{\data}{\ensuremath{\mathcal{D}}\xspace}
\newcommand{\novecs}{\ensuremath{n}\xspace}
\newcommand{\thr}{\ensuremath{\theta}\xspace}
\newcommand{\decay}{\ensuremath{\lambda}\xspace}
\newcommand{\hor}{\ensuremath{\mathcal{\tau}}\xspace}

\newcommand{\dt}{\ensuremath{\Delta t}\xspace}
\newcommand{\simt}{\ensuremath{\simop_{\dt}}\xspace}
\newcommand{\ff}[1]{\ensuremath{e^{-\decay\dt_{#1}}}\xspace}
\newcommand{\pscore}{\ensuremath{\mathtt{pscore}}\xspace}
\newcommand{\rscore}{\ensuremath{\mathtt{rem\-score}}\xspace}

\newcommand{\ltbound}{\ensuremath{\mathtt{l2bound}}\xspace}
\newcommand{\dimension}{\ensuremath{d}\xspace}
\newcommand{\reals}{\ensuremath{\mathbb{R}}\xspace}
\newcommand{\elltwo}{\ensuremath{\ell_2}\xspace}
\newcommand{\norm}[1]{\ensuremath{||{#1}||}\xspace}
\newcommand{\twonorm}[1]{\ensuremath{||{#1}||_2}\xspace}
\newcommand{\abs}[1]{\ensuremath{|{#1}|}\xspace}
\newcommand{\vect}[1]{\ensuremath{\mathbf{#1}}\xspace}
\newcommand{\vecx}{\ensuremath{\vect{x}}\xspace}
\newcommand{\xcoord}{\ensuremath{\mathrm{x}}\xspace}
\newcommand{\vecy}{\ensuremath{\vect{y}}\xspace}
\newcommand{\ycoord}{\ensuremath{\mathrm{y}}\xspace}
\newcommand{\vm}{\ensuremath{{\mathrm{vm}}}\xspace}
\newcommand{\maxsymb}{m}
\newcommand{\cmvec}{\ensuremath{\vect{\maxsymb}}\xspace}
\newcommand{\cm}{\ensuremath{\mathrm{\maxsymb}}\xspace}
\newcommand{\wvecy}{\ensuremath{\widehat{\cmvec}}\xspace}
\newcommand{\wycoord}{\ensuremath{\widehat{\cm}}\xspace}
\newcommand{\wvecydecay}{\ensuremath{\widehat{\cmvec}^{\decay}}\xspace}
\newcommand{\wycoorddecay}{\ensuremath{\widehat{\cm}^{\decay}}\xspace}
\newcommand{\prefix}[1]{\ensuremath{{#1}'}\xspace}
\newcommand{\suffix}[1]{\ensuremath{{#1}''}\xspace}
\newcommand{\timestamp}[1]{\ensuremath{t({#1})}\xspace}

\newcommand{\invind}{\ensuremath{\mathcal{I}}\xspace}
\newcommand{\ind}{\ensuremath{I}\xspace}
\newcommand{\resind}{\ensuremath{\mathcal{R}}\xspace}
\newcommand{\pointer}[1]{\ensuremath{\iota({#1})}\xspace}
\newcommand{\candidates}{\ensuremath{C}\xspace}
\newcommand{\parray}{\ensuremath{Q}\xspace}
\newcommand{\pairs}{\ensuremath{P}\xspace}
\newcommand{\sz}{\ensuremath{{\mathrm{sz}}}\xspace}
\newcommand{\rs}{\ensuremath{{\mathrm{rs}}}\xspace}
\newcommand{\bs}{\ensuremath{{\mathrm{b}}}\xspace}
\newcommand{\ps}{\ensuremath{{\mathrm{ps}}}\xspace}
\newcommand{\ds}{\ensuremath{{\mathrm{ds}}}\xspace}
\newcommand{\window}{\ensuremath{{W}}\xspace}
\newcommand{\dataset}[1]{\texttt{#1}\xspace}
\newcommand{\rcvo}{\dataset{RCV1}}
\newcommand{\webspam}{\dataset{WebSpam}}
\newcommand{\blogs}{\dataset{Blogs}}
\newcommand{\tweets}{\dataset{Tweets}}

\newcommand{\apss}{{\sc{apss}}\xspace}
\newcommand{\sssj}{{\sc{sssj}}\xspace}

\newcommand{\ap}{\ensuremath{\text{\tt{AP}}}\xspace}
\newcommand{\ii}{\ensuremath{\text{\tt{INV}}}\xspace}
\newcommand{\ltap}{\ensuremath{\text{\tt{L2AP}}}\xspace}
\newcommand{\pure}{\ensuremath{\text{\tt{L2}}}\xspace}
\newcommand{\minibatch}{\ensuremath{\text{\tt{MB}}}\xspace}
\newcommand{\x}{\ensuremath{\text{\tt{IDX}}}\xspace}
\newcommand{\minibatchx}{\ensuremath{\text{\minibatch-\x}}\xspace}
\newcommand{\streaming}{\ensuremath{\text{\tt{STR}}}\xspace}
\newcommand{\streamingx}{\ensuremath{\text{\streaming-\x}}\xspace}
\newcommand{\construct}{\ensuremath{\text{\tt{Ind\-Constr}}}\xspace}
\newcommand{\constructx}{\ensuremath{\text{\construct-\x}}\xspace}
\newcommand{\generate}{\ensuremath{\text{\tt{Cand\-Gen}}}\xspace}
\newcommand{\generatex}{\ensuremath{\text{\generate-\x}}\xspace}
\newcommand{\verify}{\ensuremath{\text{\tt{Cand\-Ver}}}\xspace}
\newcommand{\verifyx}{\ensuremath{\text{\verify-\x}}\xspace}

\newcommand{\cg}{\ensuremath{\text{\tt{CG}}}\xspace}
\newcommand{\cv}{\ensuremath{\text{\tt{CV}}}\xspace}
\newcommand{\ic}{\ensuremath{\text{\tt{IC}}}\xspace}

\SetKw{Report}{report}
\SetKwFor{During}{during}{do}{}
\SetKwInOut{Output}{output}
\SetKwInOut{Input}{input}
\SetKwComment{Comment}{// }{}
\SetKw{Andkw}{ and }
\newcommand{\define}{\ensuremath{{\,\leftarrow\,}}\xspace}
\DontPrintSemicolon

\let\emptyset\varnothing

\title{Streaming similarity self-join}
\author{
Gianmarco De Francisci Morales\\
Qatar Computing Research Institute\\
\texttt{gdfm@acm.org}
\and
Aristides Gionis\\
Aalto University\\
\texttt{aristides.gionis@aalto.fi}
}

\begin{document}
\maketitle

\enlargethispage{\baselineskip}
\begin{abstract}
We introduce and study the problem of
computing the similarity self-join in a streaming context (\sssj),
where the input is an unbounded stream of items arriving continuously.
The goal is to find all pairs of items in the stream
whose similarity is greater than a given threshold.
The simplest formulation of the problem requires unbounded memory,
and thus, it is intractable.
To make the problem feasible, we introduce the notion of \emph{time-dependent similarity:}
the similarity of two items decreases with the difference in their arrival time.

By leveraging the properties of this time-dependent similarity function,
we design two algorithmic frameworks to solve the \sssj problem.
The first one, MiniBatch (\minibatch),
uses existing index-based filtering techniques for the static version of the problem,
and combines them in a pipeline.
The second framework, Streaming (\streaming),
adds \emph{time filtering} to the existing indexes,
and integrates new time-based bounds deeply in the working of the algorithms.
We also introduce a new indexing technique (\pure),
which is based on an existing state-of-the-art indexing technique (\ltap),
but is optimized for the streaming case.

Extensive experiments show that the \streaming algorithm,
when instantiated with the \pure index,
is the most scalable option across a wide array of datasets and parameters.
\end{abstract}

%
\section{Introduction}

\begin{figure}[t]
\vspace{-\baselineskip}
\begin{center}
\includegraphics[width=0.8\columnwidth]{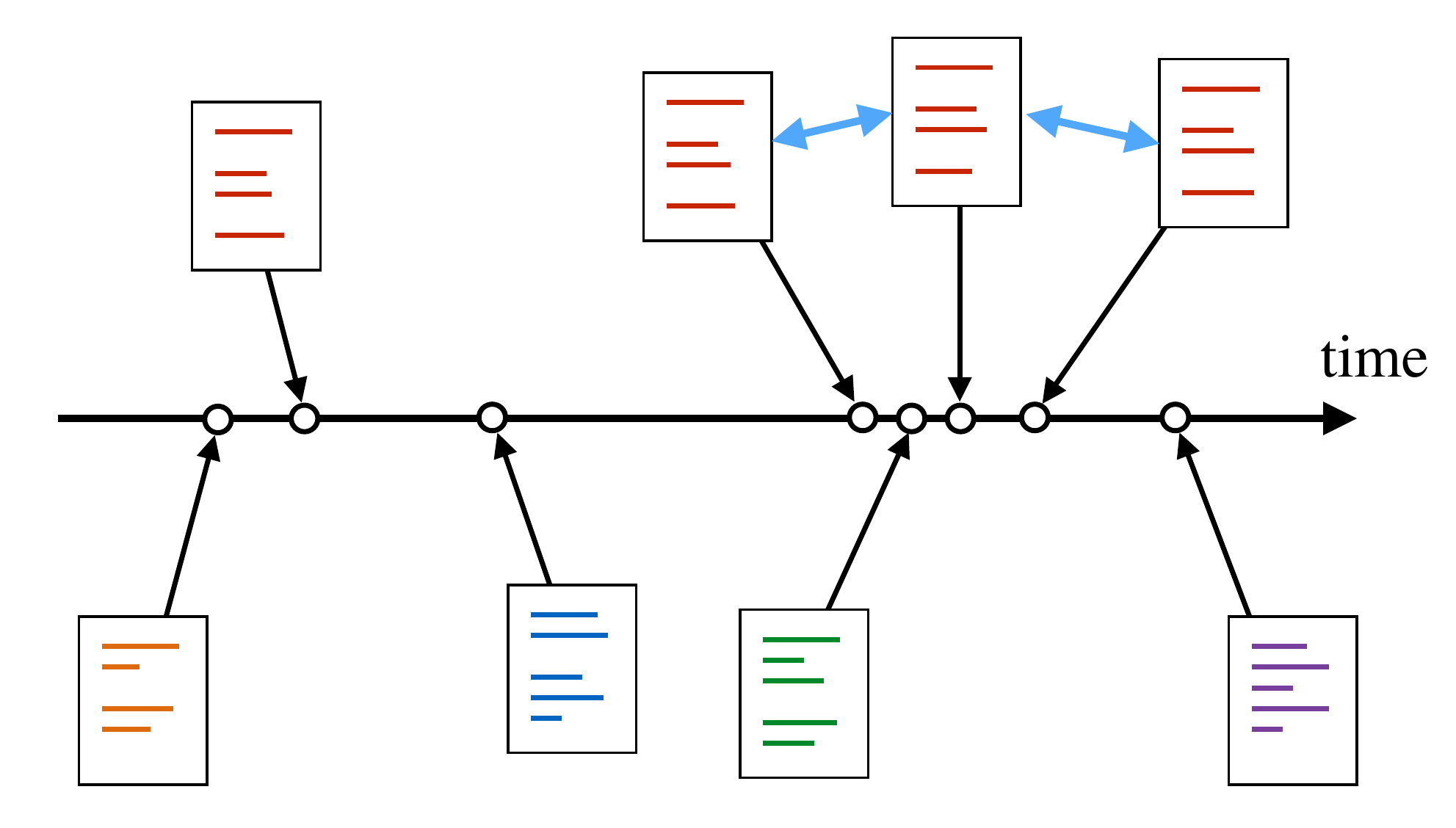}
\caption{\label{fig:illustration}
Timestamped documents arrive as stream.
The documents on top (in red) have similar content.
Among all pairs of similar documents,
we are interested in those that arrive close in time.
In the example,
out of all 4-choose-2 pairs
only two pairs are reported
(shown with blue arrows).
}
\end{center}
\vspace{-1.5\baselineskip}
\end{figure}

\emph{Similarity self-join}
is the problem of finding
\emph{all pairs} of similar objects in a given dataset.
The problem is related to the \emph{similarity join}
operator~\cite{arasu2006efficient,chaudhuri2006primitive},
which has been studied extensively in the database and data-mining communities.
The similarity self-join is an essential component
in several applications, including
plagiarism detection~\cite{cho2000finding,hoad2003methods},
query refinement~\cite{sahami2006web},
document clustering~\cite{broder1997syntactic,haveliwala2000scalable,bohm2000clustering},
data cleaning~\cite{chaudhuri2006primitive},
community mining~\cite{spertus2005evaluating},
near-duplicate record detection~\cite{xiao2011efficient},
and collaborative filtering~\cite{das2007google}.

The similarity self-join problem is inherently quadratic.
In fact, the brute-force $\bigOh(n^2)$ algorithm
that computes the similarity between all pairs
is the best one can hope for, in the worst case,
when exact similarity is required
and when dealing with arbitrary objects.
In practice, it is possible to obtain scalable algorithms
by leveraging structural properties of specific problem instances.
A typical desiderata here is to design \emph{output-sensitive} algorithms.\footnote{\url{https://en.wikipedia.org/wiki/Output-sensitive_algorithm}}

In many real-world applications
objects are represented as \emph{sparse} vectors in a high-dimensional Euclidean space, 
and similarity is measured as the dot-product of unit-normalized vectors ---
or equivalently, cosine similarity.
The similarity self-join problem
asks for all pairs of objects whose similarity is above a predefined threshold~\thr.
Efficient algorithms for this setting are well-developed,
and rely on pruning based on inverted indexes,
as well as a number of different geometric
bounds~\cite{anastasiu2014l2ap,bayardo2007scaling-up,chaudhuri2006primitive,xiao2011efficient};
details on those methods are presented in the following sections.

Computing similarity self-join is a problem
that  is germane not only for static data collections,
but also for data streams.
Two examples of real-world applications that motivate
the problem of similarity self-join in the streaming setting
are the following.

\enlargethispage{\baselineskip}
\spara{Trend detection}:
existing algorithms for trend detection in microblogging platforms,
such as Twitter,
rely on identifying hashtags or other terms
whose frequency suddenly increases.
A more focused and more granular trend-detection approach
would be to identify a set of posts, whose frequency increases,
and which share a certain \emph{fraction} of hashtags or terms.
For such a trend-detection method
it is essential to be able to find similar pairs of posts in a data stream.

\spara{Near-duplicate item filtering}:
consider again a
micro\-blog\-ging platform,
such as Twitter.
When an event occurs, users may receive
multiple near-copies of posts related to the event.
Such posts often appear consecutively in the feed of users,
thus cluttering their information stream and degrading their experience.
Grouping these near-copies or filtering them out is
a simple way to improve the user experience.

\smallskip
Surprisingly, the problem of streaming similarity self-join
has not been considered in the literature so far.
Possibly, because the similarity self-join operator in principle
requires \emph{unbounded} memory:
one can never forget a data item,
as it may be similar
to another one that comes far in the future.

To the best of our knowledge,
this paper is the first to address the similarity self-join problem
in the streaming setting.
We overcome the inherent unbounded-memory bottleneck
by introducing a temporal factor to the similarity operator.
Namely, we consider two data items similar,
only if they have 
arrived within a \emph{short time span}. 
More precisely, we define the \emph{time-dependent similarity} of two data items
to be their content-based cosine similarity
multiplied by a factor that decays exponentially with the
difference in their arrival times.
This time-dependent factor allows us to drop old items, 
as they cannot be similar to any item that arrives after a certain time horizon~\hor.
The concept is illustrated in Figure~\ref{fig:illustration}.

The time-dependent similarity outlined above
is very well-suited for both our motivating applications ---
trend detection and near-duplicate filtering:\
in both cases we are interested in identifying
items that not only are content-wise similar,
but have also arrived within a short time interval. 

Akin to previous approaches for the similarity self-join problem,
our method relies on indexing techniques.
Previous approaches use different types of \emph{index filtering}
to reduce the number of candidates
returned by the index for full similarity evaluation.
Following this terminology,
we use the term \emph{time filtering} to refer to the
property of the time-dependent similarity
that allows to drop old items from the index.

We present two different algorithmic frameworks for
the streaming similarity self-join problem,
both of which rely on the time-filtering property.
Both frameworks can be instantiated with indexing schemes
based on existing ones for the batch version of the problem.
The first framework, named MiniBatch (\minibatch),
uses existing indexing schemes off-the-shelf.
In particular, it uses two indexes in a pipeline,
and it drops the older one when it becomes old enough.
The second framework we propose, named Streaming (\streaming),
modifies existing indexing schemes so that time filtering is incorporated internally
in the index.


One of our contributions is a new index, \pure,
which combines state-of-the art bounds for index pruning from \ltap~\cite{anastasiu2014l2ap}
in a way that is optimized for stream data.
The superior performance of \pure stems from the fact
that it uses bounds that
($i$) are effective in reducing the number of candidate pairs,
($ii$) do not require collecting statistics over the data stream
(which need to be updated as the stream evolves),
($iii$) lead to very lightweight index maintenance,
while
($iv$) allowing to drop data as soon as they become too old to be similar to any item currently read.
Our experimental evaluation demonstrates
that the {\pure} index, incorporated in the {\streaming} framework,
is the method of choice for all datasets we try across a wide range of parameters.

\enlargethispage{1.5\baselineskip}
In brief, our contributions can be summarized as follows.

\begin{squishlist}
\item
We introduce the similarity self-join problem in data streams.
\item
We propose a novel time-dependent similarity measure,
which allows to forget data items when they become old.
\item
We show how to solve the proposed problem
within two different algorithmic frameworks, {\minibatch} and {\streaming}:
For {\minibatch},
any state-of-the-art indexing scheme can be used as a black box.
For {\streaming},
we adapt and extend the existing schemes to be well-suited for streaming data.
\item
We perform an extensive experimental evaluation,
and test the proposed algorithms
on datasets with different characteristics and for a large range of
parameters.
\end{squishlist}

Finally, we note that
all the code\footnote{{\small Code: \url{http://github.com/gdfm/sssj}}} we developed and
all datasets\footnote{{\small Data: \url{http://research.ics.aalto.fi/dmg/sssj-datasets/}}} used for validation
are publicly available.

\enlargethispage{\baselineskip}
\section{Related Work}
The problem of \emph{streaming similarity self-join} has not been studied previously.
Work related to this problem can be classified into two main areas:
similarity self-join
(also known as all-pairs similarity), and streaming similarity.

\spara{Similarity self-join.}
The similarity self-join problem has been studied extensively.
It was introduced by~\citet{chaudhuri2006primitive},
and many works have followed up~\citep{arasu2006efficient,awekar2009fast-matching,bayardo2007scaling-up,chaudhuri2006primitive,xiao2011efficient,xiao2008ppjoin},
including approaches that use parallel computation~\citep{afrati2011fuzzyjoin,deFrancisciMorales2010aps,baraglia2010document,vernica2010efficient,alabduljalil2013optimizing}.
Discussing all these related papers is out of the scope of this work.
The most relevant approaches are those by \citet{bayardo2007scaling-up},
which significantly improved the scalability of the previous indexing-based methods, and
by \citet{anastasiu2014l2ap},
which represents the current state-of-the-art for sequential batch algorithms.
Our work relies heavily on these two papers, and we summarize their common \emph{filtering framework} in Section~\ref{section:framework}.
For a good overview, refer to the work of~\citet{anastasiu2014l2ap}.

\spara{Streaming similarity.}
The problem of computing similarities in a stream of vectors has received surprisingly little attention so far.
Most of the work on streaming similarity focuses on time series~\citep{gao2002continually,keogh2005exact,lian2009efficient}.

\citet{lian2011similarity} study the similarity join problem for \emph{uncertain} streams of vectors 
in the sliding-window model.
They focus on uncertain objects moving in a low-dimensional Euclidean space ($\dimension \in [2,5]$).
Similarity is measured by the Euclidean distance,
and therefore the pruning techniques are very different than the ones employed in our work;
e.g., they use space-partition techniques,
and use temporal correlation of samples to improve the efficiency of their algorithm.

\citet{wang2009scalable} study the more general problem of multi-way join with expensive
user-defined functions (UDFs) in the sliding-window model.
They present a time-slicing approach,
which has some resemblance to the time-based pruning presented in this paper.
However, their focus in on distributing the computation on a cluster via pipelined parallelism,
and ensuring that the multiple windows of the multi-way join align correctly.
Differently from this paper,
their method treats UDFs as a black box.

\citet{valari2013continuous} focus on graphs, rather than vectors, and study the case of Jaccard similarity of nodes on a sliding window over an edge stream.
They apply a count-based pruning on edges by keeping a fixed-size window,
and leveraging an upper bound on the similarity within the next window.
Instead, in this work we consider time-based pruning,
without any assumption on the frequency of arrival of items in the stream.
Furthermore, the arrival of edges changes the similarity of the nodes,
and therefore different indexing techniques need to be applied;
e.g., nodes can never be pruned.

\section{Problem Statement}
\label{section:problem}

We consider data items represented as vectors in a $\dimension$-dimensional
Euclidean space $\reals^\dimension$.
In real-world applications, the dimensionality $\dimension$ is typically high,
while the data items are sparse vectors.
Given two vectors $\vecx=\langle \xcoord_1\ldots \xcoord_\dimension \rangle$ and
$\vecy=\langle \ycoord_1\ldots \ycoord_\dimension \rangle$,
their similarity $\simop(\vecx,\vecy)$
is the dot-product
\[
\simop(\vecx,\vecy) =
\dotop(\vecx,\vecy) =
\vecx\cdot\vecy =
\sum_{j=1}^{\dimension} \xcoord_j\, \ycoord_j.
\]
We assume that all vectors {\vecx} are
normalized to unit length, i.e.,
$\twonorm{\vecx} =1$,
which implies that the
dot-product
of two vectors is equal to their cosine similarity;
$\cos(\vecx,\vecy) = \dotop(\vecx,\vecy)$.

In the standard \emph{all-pairs similarity search} problem (\apss),
also known as \emph{similarity self-join},
we are given a set of vectors and a similarity threshold \thr,
and the goal is to find all pairs of vectors
$({\vecx}, {\vecy})$ for which $\simop(\vecx,\vecy)\ge\thr$.

In this paper we assume that
the input items arrive as a data stream.
Each vector $\vecx$ in the input stream
is timestamped with the time of its arrival $\timestamp{\vecx}$,
and the stream is denoted by
$\stream=\langle\ldots,(\vecx_i,\timestamp{\vecx_i}),(\vecx_{i+1},\timestamp{\vecx_{i+1}},\ldots\rangle$.

We define the similarity of two vectors
{\vecx} and {\vecy}
by considering not only their coordinates ($\xcoord_j, \ycoord_j$),
but also the difference in their arrival times in the input stream
$\dt_{\vecx \vecy} = \abs{\timestamp{\vecx}-\timestamp{\vecy}}$.
For fixed coordinates,
the larger the arrival time difference of two vectors,
the smaller their similarity.
In particular,
given two timestamped vectors {\vecx} and {\vecy}
we define their \emph{time-dependent similarity} $\simt(\vecx,\vecy)$ as
\[
\simt(\vecx,\vecy) =
\dotop(\vecx,\vecy)\, e^{- \decay\abs{\timestamp{\vecx}-\timestamp{\vecy}}},
\]
where \decay is a \emph{time-decay} parameter.

This time-dependent similarity reverts to the standard dot-product (or cosine)
similarity when $\dt_{\vecx \vecy}=0$ or $\decay=0$,
and it goes to zero as $\dt_{\vecx \vecy}$ approaches infinity,
at an exponential rate modulated by $\decay$.
As in the standard \apss\ problem,
given a similarity threshold \thr,
two vectors {\vecx} and {\vecy} are called \emph{similar}
if their time-dependent similarity is above the threshold, i.e.,
$\simt(\vecx,\vecy) \geq \thr$.

We can now define the streaming version of the \apss problem,
called the \emph{streaming similarity self-join} problem (\sssj).

\begin{problem}[\sssj]
\label{problem:streaming-apss}
Given a stream of timestamped vectors \stream,
a similarity threshold \thr,
and a time-decay factor \decay,
output all pairs of vectors
$({\vecx}, {\vecy})$ in the stream such that $\simt(\vecx,\vecy)\ge\thr$.
\end{problem}

\spara{Time filtering.}
The main challenge of computing a self-join in a data stream
is that, in principle, unbounded memory is required,
as an item can be similar with any other item that will arrive
arbitrarily far in the future.
%
By adopting the time-dependent similarity measure {\simt},
we introduce a forgetting mechanism,
which not only is intuitive from an application point-of-view,
but also allows to overcome the unbounded-memory requirement.
In particular,
since for {\elltwo}-normalized vectors $ \dotop(\vecx,\vecy)\le 1$,
we have
\[
\simt(\vecx,\vecy) =  \dotop(\vecx,\vecy) \, \ff{\vecx \vecy} \leq \ff{\vecx \vecy}.
\]
Thus, $\dt_{\vecx \vecy} > \decay^{-1} \log \thr^{-1}$
implies $\simt(\vecx,\vecy)<\thr$,
and as a result
a given vector cannot be similar to any vector that arrived
more than
\[
\hor = \frac{1}{\decay} \log \frac{1}{\thr}
\]
time units earlier.
Consequently, we can safely prune vectors that are older than \hor.
We call \hor the \emph{time horizon}.


\spara{Parameter setting.}
The \sssj problem, defined in Problem~\ref{problem:streaming-apss},
requires two parameters: a similarity threshold \thr
and a time-decay factor \decay.
The time-filtering property
suggests a simple methodology to set these two parameters:
\begin{squishlist}
\item [1.]
Select 
$\thr$
as the lowest value of the similarity
between two \emph{simultaneously-arriving}
vectors that are deemed \emph{similar}.
\item [2.]
Select 
$\hor$
as the smallest difference in arrival times
between two \emph{identical}
vectors that are deemed \emph{dissimilar}.
\item [3.]
Set $\decay = \hor^{-1} \log \thr^{-1}$.
\end{squishlist}
Setting $\thr$ and  $\hor$ in steps 1 and 2
depends on the application.

\spara{Additional notation.}
When referring to the non-streaming setting,
we consider a dataset
$\data=\{\vecx_1,\ldots,\vecx_\novecs\}$
consisting of {\novecs} vectors in $\reals^\dimension$.

Following \citeauthor{anastasiu2014l2ap},
we use the notation
$\prefix{\vecx}=\prefix{\vecx}_p=\langle\xcoord_1,..,\xcoord_{p-1},0,..,0\rangle$
and
to denote the \emph{prefix} 
of a vector \vecx. 
For a vector \vecx,
we denote by
$\vm_{\vecx}$ its maximum coordinate, by
${\Sigma}_{\vecx} = \sum_{j}\xcoord_j$ the sum of its coordinates, and by
$\abs{\vecx}$ the number of its non-zero coordinates.

Given a static dataset {\data} we use $\cm_{j}$
to refer to the maximum value along the $j$-th coordinate over all vectors in~{\data}.
All values of $\cm_{j}$ together compose the vector $\cmvec$.
Our methods use indexing schemes
which build an index incrementally, vector-by-vector.
We use $\wvecy$ to refer to the vector $\cmvec$,
restricted to the dataset that is already indexed,
and we write $\wycoord_{j}$ for its $j$-th coordinate.
Furthermore,
we use $\wvecydecay$ (and $\wycoorddecay_{j}$ for its $j$-th coordinate)
to denote a time-decayed variant of $\wvecy$,
whose precise definition is given in Section~\ref{section:framework:l2ap}.


%
%
%
%
%
%
%
%
%
%
%
%

%
\section{Overview of the approach}
\label{section:overview}

In this section we present
a high-level overview of
our main algorithms for the {\sssj} problem.

We present two different algorithmic frameworks:
{\minibatchx} (MiniBatch-\x) and {\streamingx} (Streaming-\x),
where {\x} is an indexing method for the {\apss} problem on static data.
To make the presentation of our algorithms more clear,
we first give an overview of the indexing methods for the static {\apss} problem,
as introduced in earlier
papers~\cite{anastasiu2014l2ap,bayardo2007scaling-up,chaudhuri2006primitive}.

All indexing schemes 
are based on building an \emph{inverted index}.
The index is a collection of \dimension \emph{posting lists},
$\invind=\{\ind_1,\ldots,\ind_\dimension\}$,
where the list $\ind_j$ contains all pairs
$(\pointer{\vecx},\xcoord_j)$
such that the $j$-th coordinate of vector {\vecx} is non-zero, i.e., $\xcoord_j \neq 0$.
Here $\pointer{\vecx}$ denotes a reference to vector {\vecx}.

Some methods optimize computation by not indexing the whole dataset.
In these cases, the un-indexed part is still needed in order to compute exact similarities.
Thus, we assume that a separate part of the index {\invind}
contains the un-indexed part of the dataset,
called \emph{residual direct index} \resind.

All schemes build the index incrementally,
while also computing similar pairs.
In particular,
we start with an empty index,
and we iteratively process the vectors in {\data}.
For each newly-processed vector~{\vecx},
we compute similar pairs $(\vecx,\vecy)$
for all vectors~{\vecy} that are already in the index.
Thereafter, we add
(some of) the non-zero coordinates of {\vecx} to the index.

Building the index and computing similar pairs
can be seen as a three-phase process:

\begin{squishlist}
\item [\textbf{index construction (\ic)}:]
adds new vectors to the 
index~{\invind}.
\item [\textbf{candidate generation (\cg)}:]
uses the index {\invind} to generate candidate similar pairs.
The candidate pairs may contain false positives
but no false negatives.
\item [\textbf{candidate verification (\cv)}:]
computes true similarities between candidate pairs,
and reports true similar pairs,
while dismissing false positives.
\end{squishlist}

More concretely,
for an indexing scheme {\x}
we assume that the following three primitives are available,
which correspond to the three phases outlined above:

\begin{squishlist}
\item[$(\invind, \pairs) \define \constructx(\data,\thr)$:]
Given a dataset $\data$
consisting of \novecs vectors,
and a similarity threshold \thr,
the function \constructx
returns in $\pairs=\{(\vecx,\vecy)\}$ all similar pairs $(\vecx,\vecy)$,
with $\vecx,\vecy\in\data$.
Additionally, {\constructx}
builds an index {\invind},
which can be used to find similar pairs between the vectors in {\data}
and another query vector~\vect{z}.

\item[$\candidates \define \generatex(\invind,\vecx,\thr)$:]
Given an index {\invind}, built on a dataset \data,
a vector \vecx,
and a similarity threshold~\thr,
the function {\generatex} returns
a  set of candidate vectors $\candidates=\{\vecy\}$,
which is a superset of all vectors that are similar to {\vecx}.

\item[$\pairs \define \verifyx(\invind,\vecx,\candidates,\thr)$:]
Given an index {\invind},
built on a dataset \data,
a vector \vecx,
a set of candidate vectors~$\candidates$,
and a similarity threshold~\thr,
the function {\verifyx} returns
the set
$\pairs=\{(\vecx,\vecy)\}$ of true similar pairs.
\end{squishlist}

Both proposed frameworks,
{\minibatchx} and {\streamingx},
rely on an indexing scheme \x,
and adapt it to the streaming setting by
augmenting it with the time-filtering property.
The difference between the two frameworks is in how this adaptation is done.
{\minibatchx} uses {\x} as a ``black box'':
it uses time filtering in order to
build independent instances of {\x} indexes,
and drops them when they become obsolete.
Conversely,
for {\streamingx}
we opportunely modify the indexing method {\x}
by directly applying time filtering.

The presentation of the {\streamingx} framework
is deferred to the next section,
after we present the details of the specific indexing methods {\x}
we use.

\begin{algorithm}[t]
\caption{\minibatchx\ (MiniBatch with index \x)}
\label{algorithm:minibatchx}
\Input{Data stream \stream, threshold \thr, decay \decay}
\Output{All pairs $\vecx,\vecy\in\stream$ s.t.\ $\simt(\vecx,\vecy)\ge\thr$}

$\invind \define \emptyset$\;
$t_0 \define 0; t \define 0$\;
$\hor \define \decay^{-1} \log \thr^{-1}$\;

\While{$\text{\emph{true}}$} {
	$\window \define \emptyset$\;
	$t_0 \define t_0 + \hor$\;
	\While{$t \le t_0 + \hor$} {
		$\vecx \define \texttt{read}(\stream)$\;
		$t \define t(\vecx)$\;
		$\candidates \define \generatex(\invind,\vecx,\thr)$\;
		$\pairs \define \verifyx(\invind,\vecx,\candidates,\thr)$\;
		\Report{$\mathtt{ApplyDecay} (\pairs, \decay)$}\;
		$\window \define \window \cup \{\vecx\}$\;
	}
	$(\invind, \pairs) \define \constructx(\window,\thr)$\;
	\Report{$\mathtt{ApplyDecay} (\pairs, \decay)$}\;
}
\end{algorithm}

The {\minibatchx} framework is presented below,
as we only need to know the specifications of an index {\x},
not its internal workings.
In particular, we only assume the three primitives,
{\constructx},
{\generatex}, and
{\verifyx}, are available.

The {\minibatchx} framework
works in time intervals of duration~\hor.
During the $k$-th time interval
it reads all vectors from the stream, and stores them in a buffer $\window$.
At the end of the time interval,
it invokes {\constructx} to report all similar pairs in $\window$,
and to build an index $\invind$ on $\window$.
During the $(k+1)$-th time interval,
the buffer $\window$ is reset and used to store the new vectors from the stream.
At the same time,
\minibatchx queries $\invind$ with
each vector {\vecx} read from the stream,
to find
similar pairs between {\vecx} and vectors in the previous time interval.
Similar pairs are computed by invoking the functions
\generatex\ and \verifyx.
At the end of the $(k+1)$-th time interval,
the index $\invind$ is replaced
with a new index on the vectors in the $(k+1)$-th time interval.

\minibatchx guarantees that all pairs of vectors
whose time difference is smaller than {\hor}
are tested for similarity.
Thus, due to the time-filtering property,
it returns the complete set of similar pairs.
Algorithm~\ref{algorithm:minibatchx} shows pseudocode for \minibatchx.

One drawback of \minibatchx is that
it reports some similar pairs with a delay.
In particular,
all similar pairs that span across two time intervals
are reported after the end of the first interval.
In applications that require to report similar pairs as soon as both vectors are present,
this behavior is undesirable.
Moreover, to guarantee correctness, \minibatchx needs to tests pairs of vectors
whose time difference is as large as $2\hor$, which can be pruned only
after they are reported by the indexing scheme \x, thus wasting computational power.

\section{Filtering framework}
\label{section:framework}

We now review the main indexing schemes
used for the \apss problem.
For each indexing scheme \x,
we describe its three phases 
(\ic, \cg, \cv),
and discuss how
to adapt it to the streaming setting,
giving the {\streamingx} algorithm.
One of our main contributions is the adaptation of well-known algorithms 
for the batch case to the streaming setting.
For making the paper self-contained, 
in each of the following sub-sections
we first present the indexing schemes in the static case (state of the art)
and how they are used in the \minibatch framework, 
and then we describe their adaptation to the \streaming framework (contribution of this paper).
Furthermore, 
section~\ref{section:framework:pure-l2ap} describes our improved \elltwo-based indexing scheme.

For all the indexing schemes we present, except {\ii}, we provide pseudocode.
Due to lack of space, 
and also to highlight the differences among the indexing schemes, 
we present our pseudocode using a color convention:
\noindent
$-$ {\ltap} index: all lines are included (black, \txtred{red}, and \txtgreen{green}).

\noindent
$-$ {\ap} index:
\txtred{red} lines are included; \txtgreen{green} lines excluded.

\noindent
$-$ {\pure} index:
\txtgreen{green} lines are included; \txtred{red} lines excluded.


\subsection{Inverted index}
\label{section:framework:inverted-index}


The first scheme 
is a simple \emph{inverted index} ({\ii})
with no index-pruning optimizations.
As with all indexing schemes,
the main observation is that for two vectors to be similar,
they need to have at least one common coordinate.
Thus, two similar vectors can be found together in some list~$\ind_j$.

In the {\ic} phase,
for each new vector~{\vecx},
all its coordinates $\xcoord_j$ are added to the index.
In the \cg phase,
given the query vector {\vecx},
we use the index {\invind} to retrieve candidate vectors similar to {\vecx}.
In particular, the candidates {\vecy}
are all vectors in the posting lists
where {\vecx} has non-zero coordinates.

The result of the \cg phase is the exact similarity score
between \vecx and each candidate vector \vecy.
Therefore, the \cv phase
simply applies the similarity threshold \thr
and reports the true similar pairs.

\spara{{\minibatch} framework (\minibatch-\ii):}
The functions
{\construct-\ii}, {\generate-\ii}, and {\verify-\ii},
needed to specify the {\minibatch-\ii} algorithm
follow directly from the discussion above.
Since they are rather straightforward,
we omit further details.

\spara{{\streaming} framework (\streaming-\ii):}
The description of {\ii} given above considers
the {\apss} problem, i.e., a static dataset \data.
We now consider a data stream {\stream},
where each vector {\vecx} in the stream is associated with a timestamp $\timestamp{\vecx}$.
We apply the {\ii} scheme
by adding the data items in the index in the order they appear in the stream.
The lists $\ind_j$ of the index keep pairs
$(\pointer{\vecx},\xcoord_j)$
ordered by timestamps $\timestamp{\vecx}$.
Maintaining a time-respecting order inside the lists is easy.
We process the data in the same order,
so we can can just append new vector coordinates at the end of the lists.

Before adding a new item {\vecx} to the index,
we use the index to retrieve all the earlier vectors {\vecy}
that are \dt-similar to {\vecx}.
Due to the time-filtering property,
vectors that are older than $\hor$ cannot be similar to {\vecx}.
This observation has two implications:
($i$)
we can stop retrieving candidate vectors from a list $\ind_j$
upon encountering a vector that is older than $\hor$;
and
($ii$)
when encountering such a vector,
the part of the list preceding it can be pruned away.


\begin{algorithm}[t]
\caption{\construct-\txtgreen{\pure}\txtred{\ap}}
\label{algorithm:ic-unified}
\Input{Dataset \data, threshold \thr}
\Output{Similarity index {\invind} on {\data}, and \\
set of pairs $\pairs=\{(\vecx,\vecy) \mid \simop(\vecx,\vecy)\ge\thr\}$}

$\invind \define \emptyset$\;
$\pairs \define \emptyset$\;
\ForEach{$\vecx\in\data$}{
	$\candidates \define \text{\generate-\txtgreen{\pure}\txtred{\ap}}(\invind,\vecx,\thr)$\;
	$\pairs \define \pairs \cup \text{\verify-\txtgreen{\pure}\txtred{\ap}}(\invind,\vecx,\candidates,\thr)$\;
	\txtred{$\bs_1 \define 0$}\;
	\txtgreen{$\bs_2 \define 0;  \; \bs_t \define 0$}\;
	\ForEach{$j=1\ldots\dimension \text{ \emph{s.t.} } \xcoord_j > 0$} {
		$\pscore \define \min\{\txtred{\bs_1}, \txtgreen{\bs_2}\}$\;
		\txtred{$\bs_1 \define \bs_1 + \xcoord_j \, \min\{\cm_j,\vm_{\vecx}\}$} \label{line:ic-unified:apbound}\;
		\txtgreen{$\bs_t \define \bs_t + {\xcoord_j}^{2};  \; \bs_2 \define \sqrt{\bs_t}$}\;
		\If{$\min\{\txtred{\bs_1}, \txtgreen{\bs_2}\} \ge\thr$} {
			\If{$\resind[\pointer{\vecx}] = \emptyset$} {
				$\resind[\pointer{\vecx}] \define\prefix{\vecx}_j$\;
				\txtgreen{$\parray[\pointer{\vecx}] \define \pscore$} \label{line:ic-unified:ps}\;
			}
			$\ind_j \define \ind_j\cup \{(\pointer{\vecx},\xcoord_j,||\prefix{\vecx}_j||)\}$\;
		}
	}
}
\Return $(\invind,\pairs)$\;
\end{algorithm}

\begin{algorithm}[t]
\caption{\generate-\txtgreen{\pure}\txtred{\ap}}
\label{algorithm:cg-unified}
\Input{Index \invind, vector \vecx, threshold \thr}
\Output{
Accum.\ score array $\candidates$ for candidate vectors \\
(candidate set is $\candidates=\{ \vecy\in\data \mid \candidates[\pointer{\vecy}]>0\})$
}
$\candidates \define \emptyset$\;
$\sz_1 \define \thr / \vm_{\vecx}$\;
\txtred{$\rs_1 \define \dotop(\vecx, \wvecy)$}\;
\txtgreen{$\rs_2 \define 1;  \; \rs_t \define 1$}\;
\ForEach(\CommentSty{//reverse order}){$j=\dimension \ldots 1 \text{ \emph{s.t.} } \xcoord_j > 0$} {
 	\ForEach{$(\pointer{\vecy}, \ycoord_j, ||\prefix{\vecy}_j||) \in \ind_j$}{
		$\rscore \define \min\{\txtred{\rs_1}, \txtgreen{\rs_2}\}$\;
 		\If{$|\vecy| \, \vm_{\vecy} \ge \sz_1$ \label{line:cg-unified:lowerbound}} {
 			\If{$\candidates[\pointer{\vecy}]>0 \text{ \em or } \rscore \ge\thr$ \label{line:cg-unified:upperbound}} {
				$\candidates[\pointer{\vecy}] \define \candidates[\pointer{\vecy}] + \xcoord_j \, \ycoord_j$\;
				\txtgreen{
					$\ltbound \define \candidates[\pointer{\vecy}] + ||\prefix{\vecx}_j|| \, ||\prefix{\vecy}_j||$\;
					\If{$\ltbound < \thr$\label{line:cg-unified:earlypruning}} {
						$\candidates[\pointer{\vecy}] \define 0$\;
					}
				}
 			}
 		}
 	}
 	\txtred{$\rs_1 \define \rs_1 - \xcoord_j \, \wycoord_j$}\;
	\txtgreen{$\rs_t \define \rs_t - {\xcoord_j}^{2};  \; \rs_2 \define \sqrt{\rs_t}$} \;
\label{line:cg-unified:end}
}
\Return $\candidates$;
\end{algorithm}

\begin{algorithm}[t]
\caption{\verify-\txtgreen{\pure}\txtred{\ap}}
\label{algorithm:cv-unified}
\Input{Index \invind, vector \vecx, candidate vector array $\candidates$, threshold \thr}
\Output{Set of pairs $\pairs=\{(\vecx,\vecy) \mid \simop(\vecx,\vecy)\ge\thr\}$}
$\pairs \define \emptyset$\;
\ForEach{$\vecy \text{ \emph{s.t.} } \candidates[\pointer{\vecy}] > 0$\label{line:cv-unified:start}}{
	\txtgreen{$\ps_1 \define \candidates[\pointer{\vecy}] + \parray[\pointer{\vecy}]$}\;
	$\ds_1 \define \candidates[\pointer{\vecy}] + \min\{\vm_{\vecx} \, {\Sigma}_{\prefix{\vecy}}, \vm_{\prefix{\vecy}} \, {\Sigma}_{\vecx} \}$\;
	$\sz_2 \define \candidates[\pointer{\vecy}] + \min\{|\vecx|,|\prefix{\vecy}|\} \, \vm_{\vecx} \, \vm_{\prefix{\vecy}}$\;
	\If{$\txtgreen{\ps_1 \ge \thr} \Andkw \ds_1 \ge \thr \Andkw \sz_2 \ge \thr)$\label{line:cv-unified:dpscorebound}} {
		$s \define \candidates[\pointer{\vecy}] + \dotop(\vecx, \prefix{\vecy})$\;
		\If{$(s\ge \thr)$} {
			$\pairs \define \pairs \cup \{(\vecx,\vecy,s)\}$\;
		}
	}
\label{line:cv-unified:end}}
\Return \pairs\;
\end{algorithm}

\subsection{All-pairs indexing scheme}
\label{section:framework:all-pairs}

The \ap~\citep{bayardo2007scaling-up} scheme
improves over the simple {\ii} method
by reducing the size of the index.
When using \ap, not all the coordinates of a vector {\vecx} need to be indexed,
as long as it can be guaranteed that
{\vecx} and all its similar vectors {\vecy}
share at least one common coordinate in the index {\invind}.



Similarly to the {\ii} scheme,
{\ap} incrementally builds an index by processing one vector at a time.
In the \ic phase
(function \construct-\ap, shown as Algorithm~\ref{algorithm:ic-unified},
including \txtred{red} lines and excluding \txtgreen{green}),
for each new vector {\vecx},
the algorithm scans its coordinates in a predefined order.
It keeps a score {\pscore},
which represents an upper bound on the similarity between a prefix of {\vecx}
and any other vector in the dataset.
To compute the upper bound, \ap uses the vector {\cmvec},
which keeps the maximum of each coordinate in the dataset.
As long as \pscore is smaller than the threshold {\thr},
given that the similarity of {\vecx} to any other vector cannot exceed {\thr},
the coordinates of {\vecx} scanned so far can be omitted from the index
without the danger of missing any similar pair.
As soon as {\pscore} exceeds {\thr},
the remaining suffix of {\vecx} is added to the index,
and the prefix \prefix{\vecx} is saved in the residual direct index~\resind.


In the \cg phase
(function \generate-\ap, shown as Algorithm~\ref{algorithm:cg-unified},
including \txtred{red} lines and excluding \txtgreen{green}),
\ap uses a lower bound $\sz_1$ for the size of any vector {\vecy}
that is similar to {\vecx},
so that vectors that have too few non-zero entries can be ignored
(line~\ref{line:cg-unified:lowerbound}).
Additionally, it uses a variable $\rs_1$
to keep an upper bound on the similarity
between {\vecx} and any other vector {\vecy}.
This upper bound is computed by using the residual direct index \resind
and the already accumulated dot-product,
and is updated as the algorithm processes the different posting lists.
When the upper bound becomes smaller than {\thr},
vectors that have not already been added to the set of candidates
can be ignored (line~\ref{line:cg-unified:upperbound}).
The array $\candidates$ holds the candidate vectors,
together with the partial dot-product that is due to the indexed part of the vectors.

Finally, in the \cv phase
(\verify-\ap\ function, shown as Algorithm~\ref{algorithm:cv-unified}),
we compute the final similarities by using the residual index \resind,
and report the true similar~pairs.

The streaming versions of \ap,
in both \minibatch and \streaming frameworks,
are not efficient in practice,
and thus, we omit further details.
Instead, we proceed presenting the next scheme,
the \ltap index, which is a generalization of \ap.


\subsection{L2-based indexing scheme}
\label{section:framework:l2ap}

\ltap~\citep{anastasiu2014l2ap}
is the state-of-the-art for computing similarity self-join.
The scheme uses tighter \elltwo-bounds,
which reduce the size of the index,
but also the number of generated candidates
and the number of fully-computed similarities.

The {\ltap} scheme primarily leverages the
Cauchy-Schwarz inequality, which states that
$\dotop(\vecx, \vecy) \leq ||\vecx|| \, ||\vecy||$.
The same bound applies when considering a prefix of a query vector~\vecx.
Since we consider unit-normalized vectors ($||\vecy||=1$),
\[
\dotop(\prefix{\vecx}, \vecy) \leq ||\prefix{\vecx}|| \, ||\vecy|| \leq ||\prefix{\vecx}||.
\]

The previous bound produces a tighter value for {\pscore},
which is used in the \ic phase
to bound the similarity of the vector currently being indexed,
to the rest of the dataset.
In particular,  {\ltap} sets
$\pscore = \min\{\pscore_{\ap}, ||\prefix{\vecx}||\}$.

Additionally,
the {\ltap} index stores the value of {\pscore} computed
when {\vecx} is indexed.
\ltap keeps these values in an array {\parray}, index by $\pointer{\vecx}$,
as shown in line~\ref{line:ic-unified:ps} of Algorithm~\ref{algorithm:ic-unified}.
The {\ltap} index also
stores in the index the magnitude of the prefix of each vector \vecx, up to coordinate $j$.
That is,
the entries of the posting lists are now triples of the type
$(\pointer{\vecx}, \xcoord_j, ||\prefix{\vecx}_j||)$.

Both pieces of additional information,
the array {\parray} and the values $||\vecx_j||$ stored in the posting lists,
are used during the {\cg} phase
to reduce the number of candidates.

In the \cg phase, for a given query vector {\vecx},
we scan its coordinates backwards,
i.e., in reverse order with respect to the one used during indexing,
and we accumulate similarity scores for suffixes of \vecx.
We keep a \rscore bound on the remaining similarity score,
which combines the $\rs_1$ bound used in {\ap}
and a new \elltwo-based $\rs_2$ bound that uses the prefix magnitude values
($||\prefix{\vecy}_j||$)
stored in the posting lists.
The \rscore bound is an upper bound on the similarity
of the prefix of the current query vector and any other vector in the index,
thus as long as \rscore is smaller than \thr, the algorithm can prune the candidate.

Pseudocode for the three phases of {\ltap}, \ic, \cg, and \cv,
is shown in Algorithms~\ref{algorithm:ic-unified},
\ref{algorithm:cg-unified}, and~\ref{algorithm:cv-unified}, respectively,
including both \txtred{red} and \txtgreen{green} lines.
More details on the scheme
can be found in the original paper of~\citet{anastasiu2014l2ap}.

\spara{{\minibatch} framework (\minibatch-\ltap):}
As before, the {\minibatch-\ltap} algorithm
is a direct instantiation of the generic
Algorithm~\ref{algorithm:minibatchx},
using the functions that implement the three different phases of {\ltap}.

\spara{{\streaming} framework (\streaming-\ltap):}
To describe the modifications required for adapting
the {\ltap} scheme in the streaming framework,
we need to introduce some additional notation.

First, we assume that the input is a stream {\stream} of vectors.
The main difference is that now both \cmvec and \wvecy are function of time,
as the maximum values in the stream evolve over time.
We adapt the use of the two vectors to the streaming case differently.

The vector \wvecy is used in the \cg phase to generate candidate vectors.
Given that we are looking at vectors that are already indexed (i.e., in the \emph{past}),
it is possible to apply the definition of \simt between the query vector \vecx and \wvecy.
In particular, given that the coordinates of \wvecy originate from different vectors in the index,
we can apply the decay factor to each coordinate to \wvecy separately.

More formally, let $\wvecydecay(t)$ be the \emph{worst case indexed vector} at time $t$,
that is, $\wvecydecay(t)$ is the representation of the vector in the index
that is most \dt-similar to any vector in \stream arriving at time $t$.
Its $j$-th coordinate $\wycoorddecay_j(t)$ 
is given by
\[
\wycoorddecay_j(t) =
\max_{\substack{\vecx\in\stream \\ \timestamp{\vecx}\le t}} \left\{ \xcoord_j \, e^{- \decay\abs{t-\timestamp{\vecx}}} \right\}.
\]

To simplify the notation,
we omit the dependency from time when obvious from the context,
and write simply~\wvecydecay.

An upper bound on the time-dependent
cosine similarity of any newly arrived vector {\vecx} at time $t$ can be obtained
by
\[
\rscore(t)
= \dotop(\vecx, \wvecydecay)
= \sum_{j=1}^{\dimension} \xcoord_j \, \wycoorddecay_j .
\]

Conversely, the vector \cmvec is used in the \ic phase to decide which coordinates to add to the index.
Its purpose is to guarantee that any two similar vectors in the stream share at least a coordinate.
As such, it needs to look at the \emph{future}.
In a static setting, and in \minibatch, maximum values in the whole dataset
can easily be accumulated beforehand.
However, in a streaming setting these maximum values need to be kept online.
This difference implies that when the vector \cmvec changes,
the invariant on which the \ap index is built is lost, and we need to restore it.
We call this process \emph{re-indexing}. 

At a first glance, it might seem straightforward to use \simt to decide what to index,
i.e., adding the decay factor to line~\ref{line:ic-unified:apbound} in Algorithm~\ref{algorithm:ic-unified}.
However, the decay factor causes \cmvec to change rapidly,
which in turn leads to a larger number of re-indexings.
Re-indexing, as explained next, is an expensive operation.
Therefore, we opt to avoid applying the time decay when computing the $\bs_1$ bound with \cmvec.


\begin{algorithm}[t]
\caption{\streaming-\x (Streaming with index \x)}
\label{algorithm:streamingx}
\Input{Data stream \stream, threshold \thr, decay \decay}
\Output{All pairs $\vecx,\vecy \in \stream$ s.t.\ $\simt(\vecx,\vecy)\ge\thr$}
$\invind \define \emptyset$\;
$\pairs \define \emptyset$\;
\While{$\text{\emph{true}}$} {
	\vecx \define \texttt{read}$(\stream)$\;
	$(\invind, \pairs) \define \text{\constructx-\streaming}(\invind, \vecx,\thr,\decay)$\;
	\Report{\pairs}\;
}
\end{algorithm}

\spara{Re-indexing (\ic).}
For each coordinate $\xcoord_j$ of a newly arrived vector \vecx,
one of three cases may occurr.
\begin{squishlist}
\item  [$-$]
$\xcoord_j = 0$: no action is needed.
\item  [$-$]
$0 < \xcoord_j \leq \cm_j$:
the posting list $\ind_j$ can be pruned via time filtering, as explained next.
\item  [$-$]
$\xcoord_j > \cm_j$:
the upper bound vector {\cmvec} needs to be updated with the new value just read,
i.e., $\cm_j\define\xcoord_j$.
In this case, parts of the pruned vectors may need to be re-indexed.
\end{squishlist}

When \cmvec is updated ($\xcoord_j > \cm_j$)
the invariant of prefix filtering does not hold anymore.
That is, there may be vectors similar to {\vecx} that are not matched in \invind.
To restore the invariant, we re-scan \resind
(un-indexed part of vectors within horizon {\hor}).
If the similarity of a vector $\vecy \in \resind$ and \cmvec has increased,
we will reach the threshold \thr while scanning the prefix $\prefix{\vecy}_p$
(before reaching its end).
Therefore, we just need to index those coordinates $\ycoord_{p'} < \ycoord_{j} \leq \ycoord_{p}$,
where $p'$ is the newly computed boundary.

Given that the similarity can increase only for vectors
that have non-zero value in the dimensions of {\cmvec} that got updated,
we can keep an inverted index of {\resind} to avoid scanning every vector.
We use the updated dimensions of the max vector $\cm_j$
to select the possible candidates for re-indexing,
and then scan only those from the residual index \resind.

Note that re-indexing inserts older vectors in the index.
As the lists always append items at the end,
re-indexing introduces out-of-order items in the list.
As explained in Section~\ref{section:implementation},
the loss of the time-ordered property hinders one of
the optimizations in pruning the index based on time filtering (explained next).

%

\spara{Time filtering (\cg).}
To avoid continuously scanning the index, we adopt a lazy approach.
The posting lists $\ind_j$ are (partially) sorted by time,
i.e., newest items are always appended to the tail of the lists.
Thus, a simple linear scan from the head of the lists
is able to prune the lists lazily
while accumulating the similarity in the \cg phase.
More specifically, when a new vector \vecx arrives,
we scan all the posting lists $I_j$ such that $\xcoord_j \neq 0$,
in order to generate its candidates.
While scanning the lists, we drop any item $(\pointer{\vecy}, \ycoord_j, ||\prefix{\vecy}_j||)$
such that $\abs{t(\vecx) - t(\vecy)} > \hor$.

%
%

The main loop of {\streaming-\ltap} is shown
in Algorithm~\ref{algorithm:streamingx}.
The algorithm simply reads each item {\vecx} in the stream,
adds  {\vecx} to the index, and computes the vectors that are similar to it,
by calling \construct-\ltap-\streaming
(Algorithm~\ref{algorithm:ic-unified-str},
including both \txtred{red} and \txtgreen{green} lines).

\begin{algorithm}[t]
\caption{\construct-\txtgreen{\pure}\txtred{\ap}-\streaming}
\label{algorithm:ic-unified-str}
\Input{Index \invind, vector \vecx, threshold \thr, decay \decay}
\Output{Updated similarity index \invind including \vecx, and
set of pairs $\pairs=\{(\vecx,\vecy) \mid \simt(\vecx,\vecy)\ge\thr\}$}
$\bs_1 \define +\infty; \; \bs_2 \define +\infty$\;
\txtred{$\bs_1 \define 0$}\;
\txtgreen{$\bs_2 \define 0; \bs_t \define 0$}\;
$\candidates \define \text{\generate-\txtgreen{\pure}\txtred{\ap}-\streaming}(\invind,\vecx,\thr,\decay)$\;
$\pairs \define \text{\verify-\txtgreen{\pure}\txtred{\ap}-\streaming}((\invind,\vecx,\candidates,\thr,\decay)$\;
\ForEach{$j=1\ldots\dimension \text{ \emph{s.t.} } \xcoord_j > 0$} {
	$\pscore \define \min\{\txtred{\bs_1}, \txtgreen{\bs_2}\}$\;
	\txtred{$\bs_1 \define \bs_1 + \xcoord_j \, \min\{\cm_j,\vm_{\vecx}\}$}\;
	\txtgreen{$\bs_t \define \bs_t + {\xcoord_j}^{2}; \bs_2 \define \sqrt{\bs_t}$}\;
	\If{$\min\{\txtred{\bs_1}, \txtgreen{\bs_2}\} \ge\thr$} {
		\If{$\resind[\pointer{\vecx}] = \emptyset$} {
			$\resind[\pointer{\vecx}] \define\prefix{\vecx}_j$\;
			\txtgreen{$\parray[\pointer{\vecx}] \define \pscore$} \label{line:ic-unified-str:ps}\;
		}
		$\ind_j \define \ind_j\cup \{(\pointer{\vecx}, \xcoord_j,||\prefix{\vecx}_j||)\}$\;
	}
}
\Return $(\invind, \pairs)$\;
\end{algorithm}

\begin{algorithm}[t]
\caption{\generate-\txtgreen{\pure}\txtred{\ap}-\streaming}
\label{algorithm:cg-unified-str}
\Input{Index \invind, vector \vecx, threshold \thr, decay \decay}
\Output{
Accum.\ score array $\candidates$ for candidate vectors \\
(candidate set is $\candidates=\{ \vecy\in\data \mid \candidates[\pointer{\vecy}]>0\})$
}
$\candidates \define \emptyset$\;
$\rs_1 \define +\infty; \; \rs_2 \define +\infty$\;
\txtred{$\rs_1 \define \dotop(\vecx, \wvecydecay)$}\label{line:cg-unified-str:start}\;
\txtgreen{$\rs_2 \define 1; \rs_t \define 1$}\;
\ForEach(\CommentSty{//reverse order})
	{$j=\dimension \ldots 1 \text{ \emph{s.t.} } \xcoord_j > 0$} {
 	\ForEach{$(\pointer{\vecy}, \ycoord_j, ||\prefix{\vecy}_j||) \in \ind_j \text{ \emph{s.t.} }
 		\dt_{\vecx\vecy} \leq \hor$ } {
		$\rscore \define \min\{ \txtred{\rs_1}, \txtgreen{\rs_2 \, \ff{\vecx\vecy}} \}$\label{line:cg-unified-str:upperbound}\;
		\If{$(\candidates[\pointer{\vecy}]>0 \text{ \em or } \rscore \ge\thr)$} {
			$\candidates[\pointer{\vecy}] \define \candidates[\pointer{\vecy}] + \xcoord_j \, \ycoord_j$\;
			\txtgreen{
				$\ltbound \define \candidates[\pointer{\vecy}] + ||\prefix{\vecx}_j|| \, ||\prefix{\vecy}_j|| \, \ff{\vecx\vecy}$\label{line:cg-unified-str:earlypruning}\;
				\If{$\ltbound < \thr$} {
					$\candidates[\pointer{\vecy}] \define 0$\;
				}
			}
		}
	}
 	\txtred{$\rs_1 \define \rs_1 - \xcoord_j \, \wycoorddecay_j$}\label{line:cg-unified-str:maxdecay}\;
	\txtgreen{$\rs_t \define \rs_t - {\xcoord_j}^{2}; \rs_2 \define \sqrt{\rs_t}$} \;
\label{line:cg-unified-str:end}
}
\Return $\candidates$;
\end{algorithm}

\begin{algorithm}[t]
\caption{\verify-\txtgreen{\pure}\txtred{\ap}-\streaming}
\label{algorithm:cv-unified-str}
\Input{Index \invind, vector \vecx,
candidate vector array $\candidates$, threshold \thr, decay \decay}
\Output{Set of pairs $\pairs=\{(\vecx,\vecy) \mid \simt(\vecx,\vecy)\ge\thr\}$}
$\pairs \define \emptyset$\;
\ForEach{$\vecy \text{ \emph{s.t.} } \candidates[\pointer{\vecy}] > 0$}{
	\txtgreen{$\ps_1 \define (\candidates[\pointer{\vecy}] + \parray[\pointer{\vecy}]) \, \ff{\vecx\vecy}$}\; \label{line:cv-unified-str:boundstart}
	$\ds_1 \define (\candidates[\pointer{\vecy}] + \min\{\vm_{\vecx} \, {\Sigma}_{\prefix{\vecy}}, \vm_{\prefix{\vecy}} \, {\Sigma}_{\vecx} \}) \, \ff{\vecx\vecy}$\;
	$\sz_2 \define (\candidates[\pointer{\vecy}] + \min\{|\vecx|,|\prefix{\vecy}|\} \, \vm_{\vecx} \, \vm_{\prefix{\vecy}}) \, \ff{\vecx\vecy}$\;\label{line:cv-unified-str:boundend}
	\If{$\txtgreen{\ps_1 \ge \thr} \Andkw \ds_1 \ge \thr \Andkw \sz_2 \ge \thr)$} {
		$s \define \candidates[\pointer{\vecy}] + \dotop(\vecx, \prefix{\vecy})$\;
		\If{$s\ge \thr$} {
			$\pairs \define \pairs \cup \{(\vecx,\vecy,s)\}$\;
		}
	}
\label{line:cv-unified-str:end}}
\Return \pairs\;
\end{algorithm}

\subsection{Improved L2-based indexing scheme}
\label{section:framework:pure-l2ap}

The last indexing scheme we present, {\pure},
is an adaptation of {\ltap},
optimized for stream data.

The idea for the improved index
is based on the observation that {\ltap}
combines a number of different bounds:
the bounds inherited from the {\ap} scheme
(e.g., $\bs_1$ for index construction, and $\rs_1$ for candidate generation)
and the new {\elltwo}-based bounds
introduced by~\citet{anastasiu2014l2ap}
(e.g., $\bs_2$ for index construction, and $\rs_2$ and \ltbound for candidate generation).

We have observed that in general
the {\elltwo}-based bounds are more effective than the {\ap}-based bounds.
In almost all cases the {\elltwo}-based bounds are the ones that trigger.
This observation is also verified by the results of~\citet{anastasiu2014l2ap}.
Furthermore, as can be easily seen
(by inspecting
the red and green lines of Algorithms~\ref{algorithm:ic-unified},
\ref{algorithm:cg-unified}, and~\ref{algorithm:cv-unified}),
while the {\ap} bounds use statistics of the data in the index,
the  {\elltwo}-based bounds depend only on the vector being index.
This implies that by using only the {\elltwo}-based bounds,
one does not need to maintain the worst case vector $\cmvec(t)$,
and thus, no re-indexing is required.

Thus, the {\pure} index uses only the
{\elltwo}-based bounds and discards the {\ap}-based bounds.
The static version of {\pure} (used in {\minibatch})
is shown at Algorithms~\ref{algorithm:ic-unified},
\ref{algorithm:cg-unified}, and~\ref{algorithm:cv-unified},
including the \txtgreen{green} lines and excluding the \txtred{red} ones.
The main loop for the streaming case is \construct-\pure-\streaming,
shown as Algorithm~\ref{algorithm:ic-unified-str},
including the \txtgreen{green} lines and excluding the \txtred{red} ones.

\section{Implementation}
\label{section:implementation}

This section presents additional details
and optimizations for both of our algorithmic frameworks,
{\minibatch} and {\streaming}.

\subsection{Minibatch}

Recall that in the {\minibatch} framework,
we construct the index for a set of data items
that arrive within horizon~\hor.
However,
the \ltap scheme
requires that we know not only the data used to construct the index,
but also the data used to query the index,
e.g., see Algorithm~\ref{algorithm:ic-unified}, line~\ref{line:ic-unified:apbound}
This assumption does not hold for {\minibatch},
as shown in Algorithm~\ref{algorithm:minibatchx},
as the data used to query the index arrive in the stream after
the index has been constructed.

To address this problem,
we modify the {\minibatch} framework as follows.
At any point of time we consider two windows,
$\window_{k-1}$ and $\window_{k}$,
each of size~{\hor},
where
$\window_{k}$ is the current window and $\window_{k-1}$ is the previous one.
The input data are accumulated in the current window $\window_{k}$,
and the vector {\cmvec} is maintained.
When we arrive at the end of the current window,
we compute the global vector {\cmvec}, defined over both $\window_{k-1}$ and $\window_{k}$,
by combining the {\cmvec} vectors of the two windows.
We then use the data in $\window_{k-1}$ to build the index,
and the data in $\window_{k}$ to query the index built over $\window_{k-1}$.

When moving to the next window,
$\window_{k+1}$ becomes current,
$\window_{k}$ becomes previous,
and $\window_{k-1}$ is dropped.

\subsection{Streaming}

Next we discuss implementation issues related to the streaming framework.

\spara{Variable-size lists.}
The streaming framework introduces posting lists of variable size.
The size of the posting lists not only increases,
but due to time filtering it may also decrease.
In order to avoid many and small memory (de)al\-lo\-ca\-tions,
we implement posting lists using a circular byte buffer.
When the buffer becomes full we double its capacity,
while when its size drops below 1/4 we halve it.

\spara{Time filtering.}
In the {\ii} and {\pure} schemes
it is easy to maintain the posting lists in time-increasing order.
This ordering introduces a simple optimization:
by scanning the posting lists backwards
--- from the newest to the oldest item ---
during candidate generation,
it is possible to stop scanning and truncate the posting list as soon as we find the first item over the horizon~\hor.
Truncating the circular buffer requires constant time (if no shrinking occurs).

In {\ltap} it is not possible to keep the posting lists in time order,
as re-indexing may introduce out-of-order items.
Thus, \ltap scans the lists forward and needs to go through all the expired items to prune the posting list.

\spara{Data structures.}
The residual direct index \resind
and the \parray array in \ltap and \pure
need to be continuously pruned.
To support the required operations for these data strunctures,
we implement them using a \emph{linked hash-map},
which combines a hash-map for fast retrieval, and
a linked list for sequential access.
The sequential access is the order in which the data items are inserted in the data structure,
which is also the time order.
Maintaining these data structures requires amortized constant time,
and memory linear in the number of vectors that arrive within a time interval \hor.

\spara{Applying decay-factor pruning ({\decay}).}
Central to adapting the indexing methods in the streaming setting,
is the use of the decay factor \decay for pruning.
In order to make decay-factor pruning effective
we apply the following principles.

\begin{squishlist}
\item 
[Index construction:]
As the decay factor is used to prune candidates from the data seen in the past,
decay-factor pruning is \emph{never} applied during the index-construction phase.

\item 
[Candidate generation:]
Decay-factor pruning is applied during candidate generation
when computing the \rscore bound
(line~\ref{line:cg-unified-str:upperbound} of Algorithm~\ref{algorithm:cg-unified-str}),
as well as when applying the early \elltwo pruning (line~\ref{line:cg-unified-str:earlypruning}).

For \ltap, it is also applied when computing $\rs_1$ (line~\ref{line:cg-unified-str:maxdecay}),
with a different decay factor for each coordinate
(as specified in the definition of \wvecydecay).

\item 
[Candidate verification:]
In the candidate-verification phase,
decay-factor pruning is easily applied when computing all of the bounds
(lines~\ref{line:cv-unified-str:boundstart}--\ref{line:cv-unified-str:boundend}
of Algorithm~\ref{algorithm:cv-unified-str}).
\end{squishlist}

As a final remark,
we note that the {\elltwo} bounds depend only on the
candidate and query vectors,
not on the data seen earlier in the stream.
This property makes \pure
very suitable for the streaming setting.

\section{Experimental evaluation}
\label{section:experiments}

The methods presented in the previous sections
combine four indexing schemes into two algorithmic frameworks (minibatch vs.\ streaming).
Altogether, they lead to a large number of trade-offs for
increased pruning power vs.\ index maintainance.
In order to better understand the behavior and evaluate the efficiency of each method
we perform an extensive experimental study.
The objective of our evaluation is to experimentally answer
the following questions:

\begin{squishlist}
\item [\textbf{Q1:}] Which framework performs better, \streaming or \minibatch?
\item [\textbf{Q2:}] How effective is \pure, compared to \ltap and \ii?
\item [\textbf{Q3:}] What are the effects of the parameters \decay and \thr?
\end{squishlist}

\spara{Datasets.}
We test the proposed algorithms on several real-world datasets.
\rcvo is the Reuters Corpus volume~1 dataset of newswires~\citep{lewis2004rcv1},
\webspam is a corpus of spam web pages~\citep{wang2012evolutionary},
\blogs is a collection of one month of WordPress blog posts crawled in June 2015, and
\tweets is a sample of tweets collected in June 2009~\citep{yang2011patterns}.
All datasets are available online in text format,\footnote{All datasets will be available at time of publication.}
while for the experiments we use a more compact and faster-to-read binary format;
the text-to-binary converter is also included in the source code we provide.
These datasets exhibit a wide variety in their characteristics,
as summarized in Table~\ref{tab:datasets}, and
allow us to evaluate our methods under very different scenarios.
Specifically, we see that the density of the four datasets varies greatly.
With respect to timestamps,
items in \webspam and \rcvo are assigned an artificial timestamp,
sampled from a Poisson and a sequential arrival process, respectively.
For \blogs and \tweets the publication time of each item is available,
and we use it as a timestamp.



\begin{table*}[t]
\caption{Datasets used in the experimental evaluation.
$n$:\ number of vectors;
$m$:\ number of coordinates;
$\sum{|\vecx|}$:\ number of non-zero coordinates;
$\rho = \nicefrac{\sum{|\vecx|}}{n m}$:\ density;
$\overline{|\vecx|} = \nicefrac{\sum{|\vecx|}}{n}$:\ average number of non-zero coordinates; 
and type of timestamps.}
\begin{center}
\small
\begin{tabular}{l r r r S[table-format=1.3] r r}
\toprule
\multicolumn{1}{c}{Dataset}	&	\multicolumn{1}{c}{$n$}	&	\multicolumn{1}{c}{$m$}	&	\multicolumn{1}{c}{$\sum{|\vecx|}$}	&	\multicolumn{1}{c}{$\rho$ (\%)}	&	\multicolumn{1}{c}{$\overline{|\vecx|}$}	&	\multicolumn{1}{c}{Timestamps} \\
\midrule
\webspam	&	\num{350000}		&	\num{680715}	&	\num{1305} M	&	\num{0.55}	&	\num{3728.00}	&	poisson \\ 
\rcvo	&	\num{804414}		&	\num{43001}	&	\num{61} M	&	\num{0.18}	&	\num{75.72}	&	sequential \\ 
\blogs	&	\num{2532437}		&	\num{356043}	&	\num{356} M	&	\num{0.04}	&	\num{140.40}	&	publishing date \\ 
\tweets	&	\num{18266589}	&	\num{1048576}	&	\num{173} M	&	\num{0.001}	&	\num{9.46}	&	publishing date \\ 
\bottomrule
\end{tabular}
\end{center}
\label{tab:datasets}
\vspace{-\baselineskip}
\end{table*}%

\begin{table}[t]
\caption{Fraction of the $24$ configurations of (\thr,\decay) that successfully terminate within the allowed time budget (closer to $1.00$ is better).}
\vspace{-\baselineskip}
\begin{center}
\small
\begin{tabular}{l ccc ccc}
\toprule
\multirow{2}{*}{Dataset}	&	\multicolumn{3}{c}{\minibatch}	&	\multicolumn{3}{c}{\streaming} \\
\cmidrule(lr){2-4} \cmidrule(lr){5-7}
&	\ii	&	\ltap	&	\pure		&	\ii	&	\ltap	&	\pure \\
\midrule
\webspam	&1.00&1.00&1.00&1.00&0.83&0.96 \\
\rcvo		&1.00&1.00&1.00&1.00&0.96&1.00 \\
\blogs	&0.25&0.25&0.25&1.00&0.96&1.00 \\
\tweets	&0.25&0.25&0.25&1.00&0.96&0.96 \\
\bottomrule
\end{tabular}
\end{center}
\label{tab:summary}
\end{table}%

\spara{Algorithms.}
We test the two algorithmic frameworks,
\streaming and \minibatch, with the three index variants,
\ii, \ltap, and \pure, on all four datasets.
For the similarity threshold \thr
we explore a range of values in $[0.5, 0.99]$,
which is typical for the \apss problem~\citep{anastasiu2014l2ap,bayardo2007scaling-up},
while for the time-decay factor \decay we use exponentially increasing values in the range
$[10^{-4}, 10^{-1}]$.

Our code also includes an implementation of \ap for \minibatch,
but in preliminary experiments we found it much slower than \ltap,
therefore we omit it from the set of indexing strategies under study.
Some configurations are very expensive in terms of time or memory,
and we are unable to run some of the algorithms.
We set a timeout of 3 hours for each experiment,
and we abort the run if this timout limit is exceeded, or if the JVM crashes because of lack of memory.
In all cases of failure during our experiments, \minibatch fails due to timeout,
while \streaming because of memory requirements.

\spara{Setting.}
We run all the experiments on a computer with an Intel Xeon CPU E31230 @ 3.20$\,$GHz with 8$\,$MB of cache,
32$\,$GB of RAM (of which 16$\,$GB allocated for the JVM heap), and a 500$\,$GB SATA disk.
All experiments run sequentially on a single core of the CPU.
We warm up the disk cache for each dataset by performing one run of the algorithm before taking running times.
Times are averaged over three runs.

\subsection{Results}
\label{section:experiments:results}

\spara{Q1 ({\minibatch} vs.\ {\streaming}):}
Of the four datasets we use,
\minibatch successfully runs with all configurations on only two of them (\rcvo and \webspam).
As can be seen from Table~\ref{tab:datasets}, \blogs and \tweets are the largest datasets,
and \minibatch is not able to scale to these sizes.
On the other hand, \streaming is able to run with all configurations on all four datasets.
The overhead of \minibatch becomes too large when the horizon \hor becomes small enough,
as a new index needs to be initialized too frequently.
Table~\ref{tab:summary} shows a summary of the outcomes of the experimental evaluations for the various configuration.
Therefore, we restrict this comparison to \rcvo and \webspam.

Most of the time of \streaming and \minibatch, according to our profiling results,
is spent in the \cg phase while scanning the posting lists.
Hence, we first compare the algorithms on the number of posting entries traversed.
Figure~\ref{fig:entries-ratio} shows that \streaming usually does less total work,
and traverses around $65\%$ of the entries traversed by \minibatch.
The figure shows the ratio for \pure,
but other indexing strategies show the same trend (plots are omitted for brevity).

We now turn our attention to running time.
Figures~\ref{fig:time-rcv1} and~\ref{fig:time-webspam}
compare the running time of \streaming and \minibatch
on \rcvo and \webspam for all configuration on which both are able to run.
The decay factor \decay varies with the columns of the grid,
while the indexing scheme with the rows.
The two datasets present a different picture.

Clearly, for \rcvo (Figure~\ref{fig:time-rcv1}) algorithm \streaming is faster than \minibatch in most cases.
More aggressive pruning reduces the difference,
while in some cases with low \thr (useful for recommender systems) algorithm \streaming
can be up to $4$ times faster than~\minibatch.
\ltap is the exception, and for smaller \hor its performance is subpar, as detailed next.

Conversely, for \webspam the \minibatch algorithm
has the upper hand in most cases,
especially for larger decay factors (i.e., shorter horizons).
The different behavior is caused by the higher density of \webspam,
which has an average number of non-zero coordinates per vector,
which is almost two orders of magnitude larger than \rcvo (see Table~\ref{tab:datasets}).
For \streaming, this high density renders the lazy update approach inefficient,
as a large number of posting lists need to be updated and pruned for each vector,
especially for short horizons.
\minibatch can simply throw away old indexes, rather than mending them.

Summarizing our findings related to the {\minibatch}-vs.-{\streaming} study,
we conclude that \streaming is more efficient in most cases,
as it is able to run on all datasets, while \minibatch becomes too slow for larger datasets.
For very specific conditions of high density, \minibatch has a small advantage on \streaming.

\begin{figure}[t]
\begin{center}
\includegraphics[width=\columnwidth]{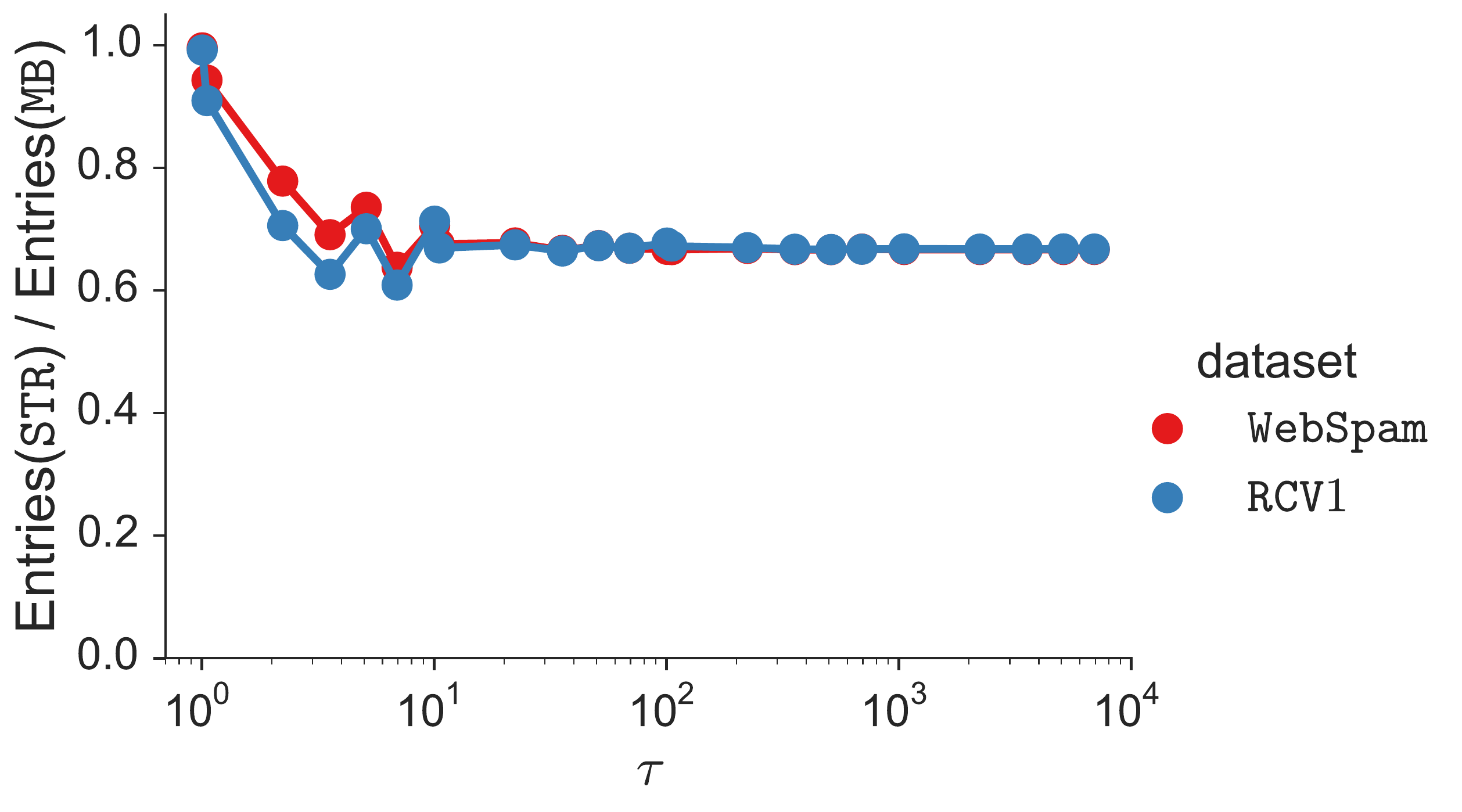}
\caption{\label{fig:entries-ratio}
Ratio of index entries traversed during \cg for \streaming compared to \minibatch.
For small \hor, the ratio tends to one, while for larger horizons \hor,
\streaming needs to traverse only $65\%$ of the entries compared to \minibatch.}
\end{center}
\vspace{-\baselineskip}
\end{figure}

\begin{figure*}[t]
\begin{center}
\includegraphics[width=\textwidth]{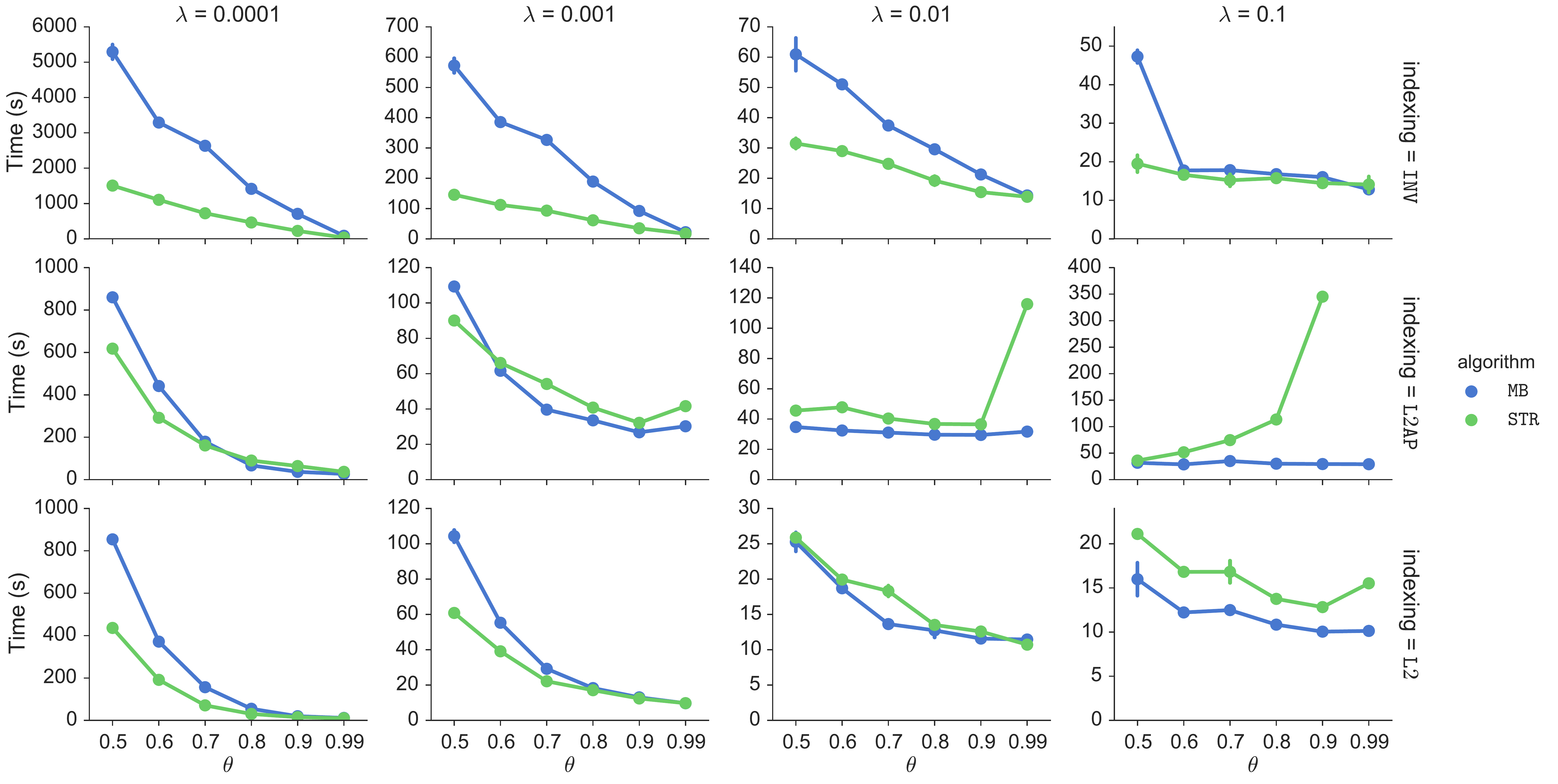}
\caption{Time taken by the \minibatch and \streaming algorithms
as a function of the similarity threshold \thr, on the \rcvo dataset.}
\label{fig:time-rcv1}
\end{center}
\vspace{-1.2\baselineskip}
\end{figure*}

\begin{figure*}[t]
\begin{center}
\includegraphics[width=\textwidth]{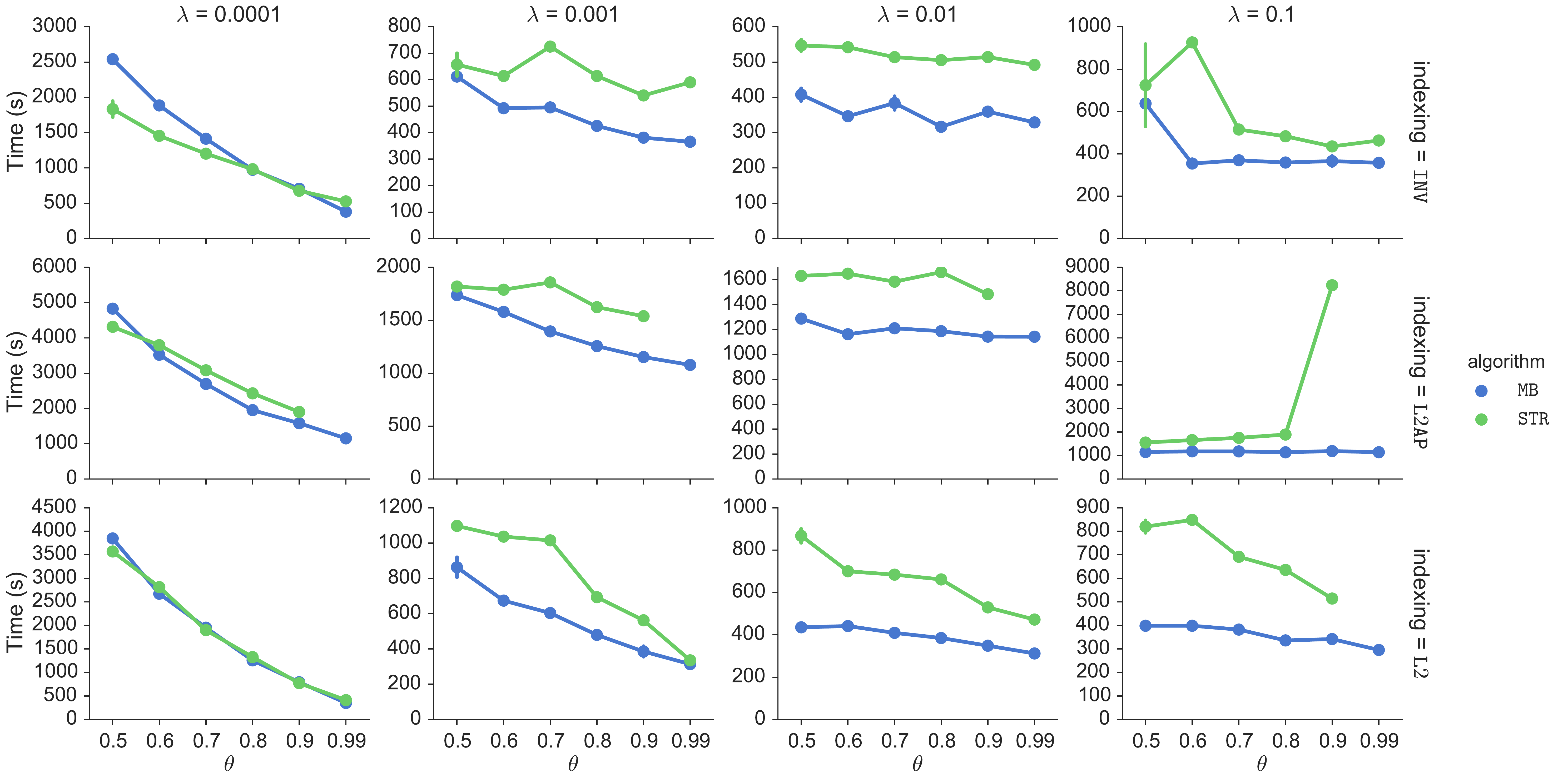}
\caption{Time taken by the \minibatch and \streaming algorithms
as a function of the similarity threshold \thr, on the \webspam dataset.}
\label{fig:time-webspam}
\end{center}
\vspace{-1.2\baselineskip}
\end{figure*}




\begin{figure*}[t]
\begin{center}
\includegraphics[width=\textwidth]{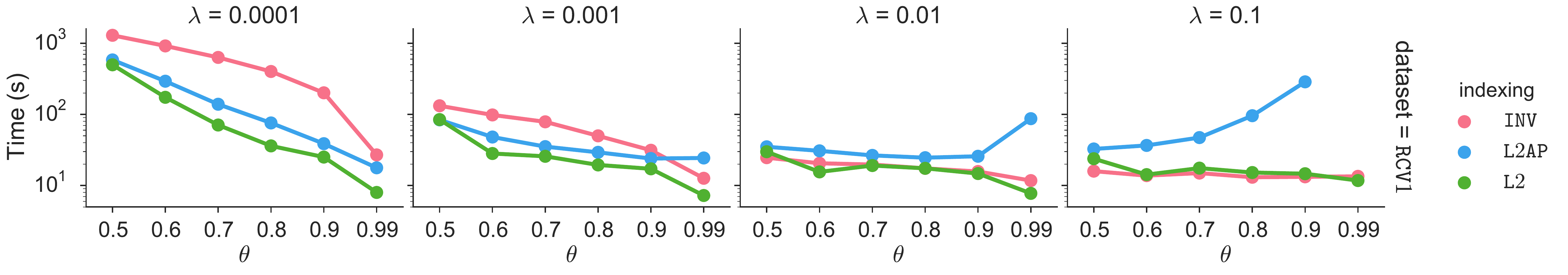}
\caption{Time taken by \streaming using different indexes
as a function of the similarity threshold \thr, for the \rcvo dataset.}
\label{fig:streaming-indexing-comparison-time}
\end{center}
\vspace{-1.2\baselineskip}
\end{figure*}

\begin{figure*}[t]
\begin{center}
\includegraphics[width=\textwidth]{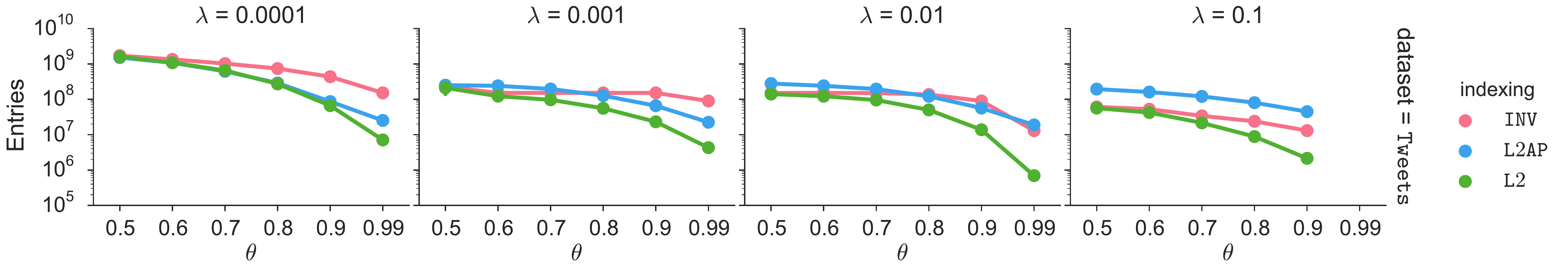}
\caption{Entries traversed by \streaming using different indexes
as a function of the similarity threshold \thr, for the \tweets dataset.}
\label{fig:streaming-indexing-comparison-entries}
\end{center}
\vspace{-1.2\baselineskip}
\end{figure*}

\spara{Q2 (Indexing schemes):}
Now that we have established that \streaming  scales better than \minibatch,
we focus on the former algorithm:
the rest of the experiments are performed on \streaming.
We turn our attention to the effectiveness of pruning, and compare \pure with \ii and \ltap.
For brevity, we show results for \rcvo only.
The other datasets follow the same trends.


Figure~\ref{fig:streaming-indexing-comparison-time} shows the comparison in time.
Several interesting patterns emerge.
First, \pure is almost always the fastest indexing scheme.
Second, \ii works well for short horizons,
where the overheads incurred by pruning are not compensated by a large reduction in candidates.
Instead, for larger horizons (low \thr and \decay)
the gains of pruning are more pronounced, thus making \ii a poor choice.
Finally, \ltap is slightly slower than \pure in most cases,
and much slower when the horizon is short.
Even though \pure uses a subset of the bounds of \ltap,
the overhead of computing the \ap bounds offsets any possible gain in pruning.

In addition, \ltap may need to re-index residuals of vectors,
and a shorter horizon causes more frequent re-indexings.
Therefore, the two rightmost plots show an upward trend in \ltap,
which is due to the re-indexing overhead.
The overhead is so large that we were not able to run \ltap for $\thr=0.99$ and $\decay=0.1$.
While there are possible practical workarounds
(e.g., use a more lax bound to decrease the frequency of re-indexing),
\pure always provides a better choice and does not require tinkering with the algorithm.

In general, a shorter horizon makes other pruning strategies less relevant compared to time filtering,
as can be seen from the progressive flattening of the curves from left to right as \decay increases.

Figure~\ref{fig:streaming-indexing-comparison-entries} shows the comparison
in the number of index entries traversed.
Clearly, \ii has usually the largest amount
due to the absence of pruning. 
The relative effectiveness of pruning for \pure is almost constant
throughout the range of \decay,
however the number of entries traversed decreases,
and so does the importance of filtering the index (compared to the time filtering).
This is in line with the behavior exhibited
in Figure~\ref{fig:streaming-indexing-comparison-time}
where the difference in time decreases for larger~\decay.

Interestingly, \ltap starts very close to \pure,
but as the horizon shortens
the number of entries traversed increases significantly,
evens surpassing \ii in the rightmost plot.
This result shows that \pure does not lose much in terms of pruning power,
despite not using the bounds from \ap.
Furthermore, the optimization to the implementation of time filtering
described in Section~\ref{section:implementation} (backwards posting-list scanning)
is not applicable.
The reason is that the \ap bounds required during indexing are data-dependent,
which leads to re-indexing and therefore there is no guarantee that the posting lists are sorted by time.
Therefore, \ltap ends up traversing more entries than \pure.

Note that the last points on the rightmost plot of
Figure~\ref{fig:streaming-indexing-comparison-entries}
($\thr=0.99$ and $\decay=0.1$) are not shown because their value is zero,
which cannot be represented on a logarithmic axis.
Indeed, in this configuration $\hor=0.1$ which is smaller than any timestamp delta, so the index gets continuously pruned before bing traversed, and therefore the output is empty.

Similar trends are observed for the number of candidates generated and
the number of full similarities computed.
Those results are omitted due to space constraints.

%

%

\begin{figure*}[t]
\begin{center}
\includegraphics[width=\textwidth]{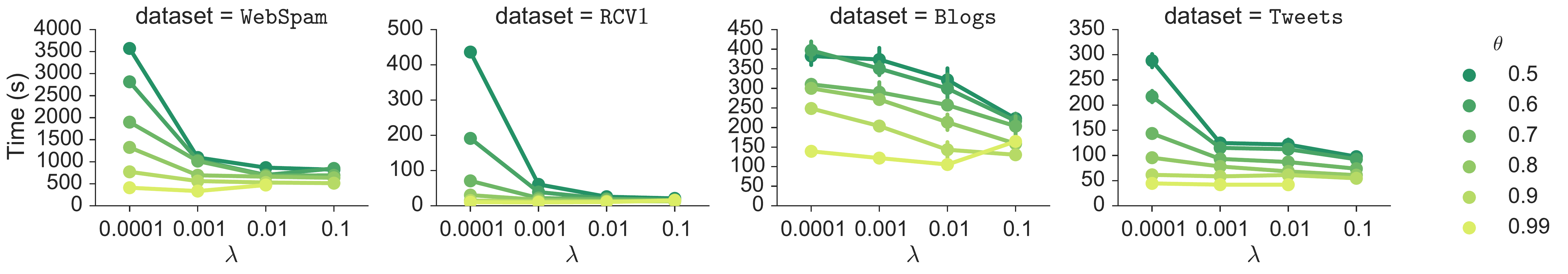}
\caption{Time taken by \streaming using the \pure index
as a function of the decay factor \decay,
for different similarity thresholds~\thr.}
\label{fig:parameters-lambda-time}
\end{center}
\vspace{-1.2\baselineskip}
\end{figure*}

\begin{figure*}[t]
\begin{center}
\includegraphics[width=\textwidth]{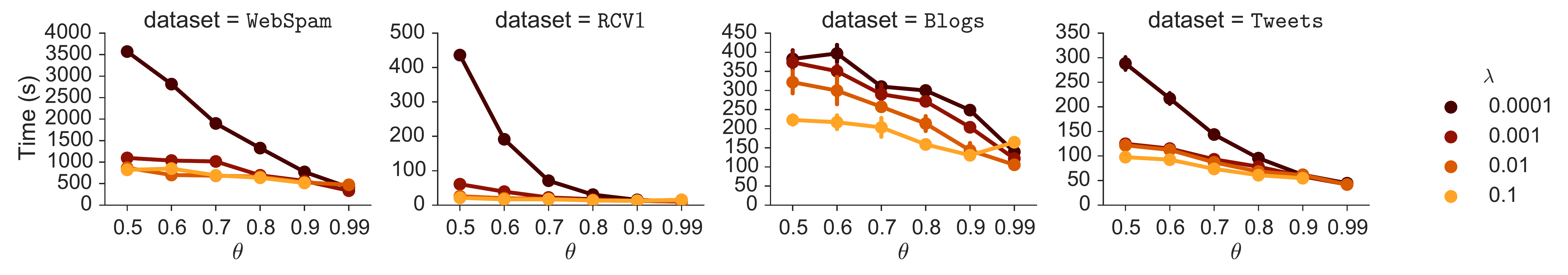}
\caption{Time taken by \streaming using the \pure index
as a function of the similarity thresholds \thr,
for different decay factors \decay.}
\label{fig:parameters-theta-time}
\end{center}
\vspace{-1.2\baselineskip}
\end{figure*}

\spara{Q3 (Parameters \thr and \decay):}
Next we study the effect of the parameters of our methods:
the decay factor \decay and the similarity threshold \thr.
Hereinafter, we present results only for \pure,
for which we have established that it is an effective and efficient pruning scheme.

Figure~\ref{fig:parameters-lambda-time} shows the effect of the decay factor \decay.
Increasing the decay factor decreases the computation time for all datasets.
However, the decrease is more marked for lower threshold \thr, and
flattens out quickly for higher values of \decay.

Figure~\ref{fig:parameters-theta-time} shows the effect of the similarity threshold \thr.
The pattern is similar to the previous figure, with the roles of \thr and \decay reversed.
Increasing the threshold decreases the computation time for all datasets, but is more marked for lower \decay, and flattens out quickly.

By combining the insights from the previous two figures, we can infer that both parameters
jointly affect the computation time.
The reason is that time is mainly affected by time filtering,
which is the most effective pruning strategy in this setting,
and its effect depends on the horizon \hor,
which jointly depends on the other two main parameters.

To explore this hypothesis, we perform a linear regression of the computation time on the value of the horizon \hor.
Figure~\ref{fig:parameters-tau-time} shows that the computation time is roughly a linear function of \hor.
In addition, from the inclination of the regression line
it is clear that \webspam is an outlier when compared to the other datasets, as previously mentioned.

\section{Conclusion}
We introduced the problem
of computing the similarity self-join in data streams.
Our approach relies on incorporating a forgetting
factor in the similarity measure so as to
be able to prune data items when they become old enough.
Given the new definition of time-dependent similarity,
we developed two algorithmic frameworks
that incorporate existing indexing techniques
for computing similarity self-join on static data,
and we extended those techniques to the streaming case.
We explored several different combinations of bounds used for index pruning,
which, in the context of streaming data,
lead to interesting performance trade-offs.
Our extensive analysis allows to understand better these trade-offs,
and consequently, to design an index that is optimized for streaming data.
Our analysis indicates that the {\streaming} algorithm
equipped with the {\pure} index is the most scalable
and robust across all datasets and configurations.

One promising direction for future
work is to experiment with dimension-ordering strategies
and evaluate the cost-benefit trade-off of maintaining a dimension ordering.
Other directions include
applying the developed techniques in real-world applications
for filtering near-duplicate items in data streams,
as well as extending our model for different
definitions of time-dependent similarity.

\begin{figure}[t]
\begin{center}
\includegraphics[width=\columnwidth]{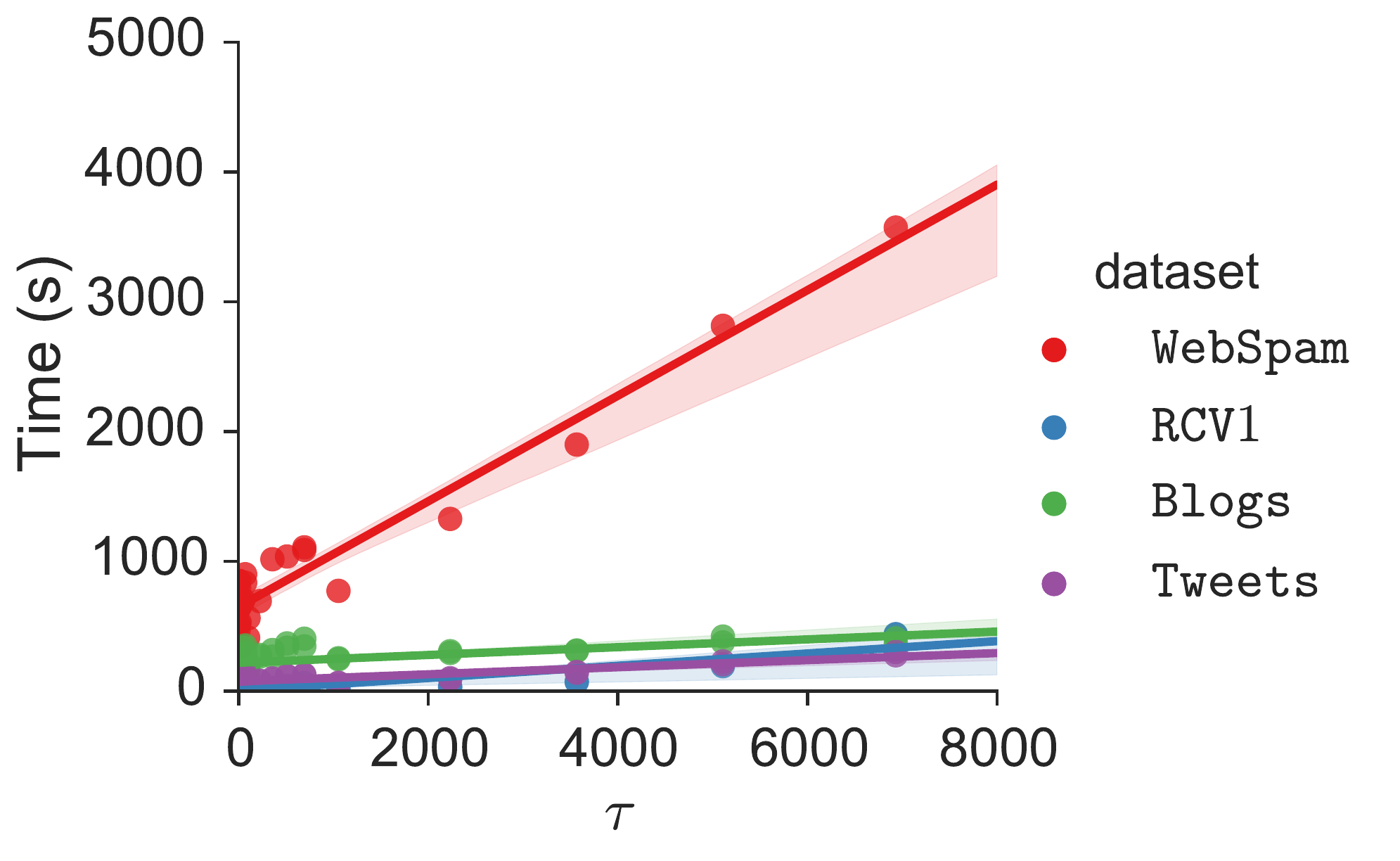}
\caption{Time taken by the \streaming algorithm when using the \pure index
as a function of the time horizon \hor.}
\label{fig:parameters-tau-time}
\end{center}
\vspace{-\baselineskip}
\end{figure}

\bibliography{sssj}
\bibliographystyle{nourlabbrvnat}

\appendix
\section{Correctness}
It is clear that our algorithms do not produce any false positive, as there is always a final check of the actual similarity threshold before the final output of any pair.

For false negatives, proof of correctness follows from the definition of time-based similarity and the Cauchy-Schwarz inequality.

\begin{theorem}
\cg phase is safe, i.e., pruning does not miss any similar pair.
\end{theorem}
\begin{proof}
\[
\simt(\vecx,\vecy) = 
\dotop(\vecx,\vecy) \, e^{- \decay\abs{\timestamp{\vecx}-\timestamp{\vecy}}} \leq
\norm{\vecx} \, \norm{\vecy} \, \ff{\vecx \vecy} \leq
\ff{\vecx \vecy}
\]
The first inequality supports the pruning rules at Algorithm~\ref{algorithm:cg-unified-str} Lines 10-12, and Algorithm~\ref{algorithm:cv-unified-str} Lines 3 and 6 (green code).
The second inequality supports the pruning rule at Algorithm~\ref{algorithm:cg-unified-str} Line 7 (green code).

In particular, for any index $j \in [1, \dimension]$, we have
\[
\simt(\vecx,\vecy) \leq
( \norm{\prefix{\vecx}_j} \, \norm{\prefix{\vecy}_j} +
\norm{\suffix{\vecx}_j} \, \norm{\suffix{\vecy}_j} ) \, \ff{\vecx \vecy}.
\]
In the algorithm, the coordinate $j$ corresponds to any of the indexed dimensions.
\end{proof}

The proof of correctness relies on the Cauchy-Schwarz inequality, and the fact that the time-based similarity adds a forgetting \emph{factor}.
Therefore, the factor can be carried together with the proof.
The forgetting factor actually makes the bound tighter for the streaming setting, as $\ff{\vecx \vecy} \leq 1$.

\end{document}